\documentclass[12pt,a4paper]{article}

\usepackage{amsmath}
\usepackage{amssymb}
\usepackage{mathtools}
\usepackage{tikz}
\usepackage{graphicx}
\usepackage{epsfig}
\usepackage[outdir=./]{epstopdf}
\usepackage{amsthm}
\usepackage{fullpage}

\usepackage[textwidth=1.9cm]{todonotes}

\theoremstyle{plain}
\newtheorem{thm}{Theorem}[section]
\newtheorem{cor}[thm]{Corollary}
\newtheorem{lem}[thm]{Lemma}
\newtheorem{obs}[thm]{Observation}

\theoremstyle{definition}
\newtheorem{dfn}[thm]{Definition}
\newtheorem{exm}[thm]{Example}

\newcommand{\lmulset}{\{\!\!\{}
\newcommand{\rmulset}{\}\!\!\}}

\DeclareMathOperator{\CFI}{CFI}
\DeclareMathOperator{\supp}{supp}
\DeclareMathOperator{\Iso}{Iso}
\DeclareMathOperator{\Aut}{Aut}
\DeclareMathOperator{\Sub}{Sub}
\DeclareMathOperator{\co}{co}
\DeclareMathOperator{\rank}{rank}
\DeclareMathOperator{\Sym}{Sym}

\newcommand{\versone}{\textup{I}}
\newcommand{\verstwo}{\textup{II}}
\newcommand{\versthree}{\textup{III}}

\usepackage[colorlinks]{hyperref}
\definecolor{lightblue}{rgb}{0.5,0.5,1.0}
\definecolor{darkred}{rgb}{0.5,0,0}
\definecolor{darkgreen}{rgb}{0,0.5,0}
\definecolor{darkblue}{rgb}{0,0,0.5}

\hypersetup{colorlinks,linkcolor=darkred,filecolor=darkgreen,urlcolor=darkred,citecolor=darkblue}

\author{Jendrik Brachter\\
TU Kaiserslautern\\
\and
Pascal Schweitzer\\
TU Kaiserslautern\\
}

\begin{document}
\title{On the Weisfeiler-Leman Dimension of Finite Groups}
\maketitle

\begin{abstract}
In comparison to graphs, combinatorial methods for the isomorphism problem of finite groups are less developed than algebraic ones.
To be able to investigate the descriptive complexity of finite groups and the group isomorphism problem, we define the Weisfeiler-Leman algorithm for groups. In fact we define three versions of the algorithm. In contrast to graphs, where the three analogous versions readily agree, for groups the situation is more intricate. For groups, we show that their expressive power is linearly related. We also give descriptions in terms of counting logics and bijective pebble games for each of the versions. 

In order to construct examples of groups, we devise an isomorphism and non-isomorphism preserving transformation from graphs to groups.
Using graphs of high Weisfeiler-Leman dimension, we construct highly similar but non-isomorphic groups with equal~$\Theta(\sqrt{\log n})$-subgroup-profiles, which nevertheless have Weisfeiler-Leman dimension 3. These groups are nilpotent groups of class 2 and exponent~$p$, they agree in many combinatorial properties such as the combinatorics of their conjugacy classes and have highly similar commuting graphs.

The results indicate that the Weisfeiler-Leman algorithm can be more effective in distinguishing groups than in distinguishing graphs based on similar combinatorial constructions.
\end{abstract}

\section{Introduction}
The notion of isomorphisms between finite groups remains one of the most basic concepts of group theory for which we do not have efficient algorithmic tools. The algorithmic Group Isomorphism Problem formalizes the task of deciding whether two given (finite) groups are isomorphic, but in fact, we do not understand its complexity. We have neither a polynomial time algorithm for testing isomorphism, nor complexity theoretic evidence indicating to us that the problem is not polynomial time solvable.
 Considering groups of order $n$, a simple approach, attributed to Tarjan in \cite{Mi78}, is to pick a small generating set in one of the groups and to check for all possible images of the generators in the other group, whether the partial map extends to an isomorphism.
This approach gives us a worst-case runtime of $n^{\log t+\mathcal{O}(1)}$ where~$t$ is the size of the generating set. Since every group of order~$n$ has a generating set of size at most~$\log n$, this yields $n^{\log n+\mathcal{O}(1)}$ in the worst case. Despite decades of active research this bound has seen only slight improvements for the general case. In fact, Rosenbaum~\cite{DBLP:journals/corr/abs-1304-3935} was able to improve it to $n^{1/2\log n+\mathcal{O}(1)}$. (See~\cite{DBLP:journals/corr/Luks15} and~\cite{DBLP:journals/tcs/RosenbaumW15} for related discussions on isomorphism of~$p$-groups and solvable groups). For various classes of groups, better bounds are known (see further related work). However, even very limited classes of groups provide hard cases for isomorphism testing. One of the most prominent classes in this context is formed by the groups of prime exponent $p$ and nilpotency class $2$. Such groups possess a lot of extra structure, but despite this and despite a large body of research into this structure, even for this limited class, no better general bound has been proven. In fact, this class seems to be at the core of the problem. However, a formal reduction to this or a similar class is not known.

While there exists a vast collection of algebraic methods and heuristics for tackling the group isomorphism problem (see further related work), complexity theoretic and combinatorial aspects seem to be less developed. 
For example, in 2011, Timothy Gowers asked on Lipton's blog~\cite{gowers} whether there is an integer~$m$ such that the isomorphism class of each finite group is determined by their~$m$-subgroup-profile. Here the~$m$-subgroup-profile (or~$m$-profile) is the multiset of (isomorphism types of) $m$-generated subgroups. Glauberman and Grabowski gave a negative answer by constructing pairs of non-isomorphic groups with the same $\Theta(\sqrt{\log n})$-profiles~\cite{GlaubermanGrabowski}. Subsequently, Wilson constructed many examples of exponent $p$ and nilpotency class $2$ groups which agree in various invariants. In particular they have the same~$\Theta(\log n)$-profiles \cite{WilsonThreshold}, which is best possible.

The observation that combinatorial aspects of the group isomorphism problem are less developed is surprising since, for the related graph isomorphism problem, historically, it has been the other way around. 
Indeed, for graph isomorphism testing, combinatorial approaches are well-developed and often successful, yet their limits have been firmly established. One of the most important tools in this scope is the
Weisfeiler-Leman algorithm. The $k$-dimensional Weisfeiler-Leman algorithm ($k$-WL) iteratively classifies $k$-tuples of vertices of a graph in terms of how they are related to other vertices in the graph. It provides an effective invariant for graph-non-isomorphism (see e.g.~\cite{MR0543783,DBLP:journals/combinatorica/CaiFI92}). Moreover, $k$-WL can be implemented to run in time $\mathcal{O}(n^{k+1}\log n)$ where $n$ is the number of vertices (see~\cite{MR1060782,1907.09582}). For fixed~$k$, the algorithm is only a partial isomorphism test, in that it can distinguish certain pairs of non-isomorphic graphs, but not all of them.
 A graph is said to have WL-dimension at most~$k$, if~$k$-WL distinguishes the graph from every non-isomorphic graph. For many important classes of graphs the WL-dimension has been shown to be bounded; examples include planar graphs \cite{DBLP:conf/stoc/Grohe00,DBLP:journals/jacm/KieferPS19} for which even~$3$ suffices and more generally classes defined by forbidden minors~\cite{MR3729479}. On the other hand, Cai, F\"urer and Immerman constructed an infinite family of graphs, for which the WL-dimension is linear in the number of vertices, and thus unbounded~\cite{DBLP:journals/combinatorica/CaiFI92}. Higher dimensional versions of $k$-WL also appear in Babai's breakthrough result putting graph-isomorphism in quasi-polynomial time~\cite{DBLP:conf/stoc/Babai16}.
 
There is a deep and well-understood connection between $k$-WL and the expressiveness in the logic $\mathbf{C}^{k+1}$, the extension of the $(k+1)$-variable fragment of first order logic on graphs with counting quantifiers~\cite{DBLP:journals/combinatorica/CaiFI92}. For example, two graphs can be distinguished by~$k$-WL exactly if there is a formula in~$\mathbf{C}^{k+1}$ that distinguishes the graphs. Therefore, in some well-defined sense, the~$k$-WL algorithm is universal in that it simultaneously checks all combinatorial properties in an input graph expressible in the aforementioned logic.

\textbf{Contribution.} The first aim of this paper is to introduce Weisfeiler-Leman-type algorithms and the notion of a WL-dimension for groups analogous to the graph case. While at first sight it seems straightforward to do so, it turns out that various concepts that coincide when applied to graphs (potentially) disagree when applied to groups. Specifically, we define three natural but different versions of a Weisfeiler-Leman dimension. One of them is based on a natural logic for groups while another is natural when taking an algorithmic viewpoint. The third version comes from natural translation of groups into graphs in an isomorphism and non-isomorphism preserving manner. Still, we give descriptions in terms of counting logics and bijective pebble games for each of the versions. 
A core reason why the different versions arise is that the correspondence between the various concepts arising in this context (specifically logics, algorithms and pebble games) is not as clean as for graphs.
However, we argue that the definition is robust after all: we prove that the Weisfeiler-Leman dimensions of the different versions are linearly related.
Overall, we obtain a family of algorithms that is similarly universal in checking combinatorial properties as in the graph case.
For example, it is easy to see that the~$k$-WL algorithm implicitly computes the~$k$-profile of groups.
In particular, abelian groups are completely identified already by the least powerful of the algorithms.

The second aim of this paper is to understand when and how groups are characterized by their combinatorial properties. On the one hand this addresses the question whether combinatorial methods can solve the Group Isomorphism Problem. On the other hand it provides a way of quantifying similarity of non-isomorphic groups. Specifically, we construct pairs of arbitrarily large non-isomorphic groups that agree with respect to many isomorphism invariants but can still be distinguished with the $3$-dimensional WL-algorithm. More precisely these groups are of nilpotency class $2$ and prime exponent $p$. They are non-isomorphic but have the same~$\Theta(\sqrt{\log n})$-profile\footnote{In a previous version of the paper we calculated the orders of the group incorrectly and therefore claimed the statement for~$\Theta({\log n})$-profiles rather than~$\Theta(\sqrt{\log n})$-profiles.}. They also have highly similar commuting graphs.
\begin{thm}
For infinitely many~$n$ there exist pairs of non-isomorphic groups order~$n$ with bounded Weisfeiler-Leman dimension which 
\begin{itemize}
\item have equal~$\Theta(\sqrt{\log n})$-profiles,
\item have commuting graphs that are indistinguishable for the  $\mathcal{O}(\sqrt{\log n})$-dimensional Weisfeiler-Leman algorithm (for graphs),
\item are of exponent $p$ and nilpotency class $2$, and
\item have equal sizes of conjugacy classes.
\end{itemize}
\end{thm}

The theorem shows that the Weisfeiler-Leman algorithm can be more effective in distinguishing groups than in distinguishing graphs even when they are based on similar combinatorial constructions. 
The proof that the WL-dimension is low intuitively indicates that the ability to fix products of elements appears to be related to the ability to fix sets of elements and how to exploit this.

In comparison to the previous constructions mentioned above, our construction has the advantage that it is of a purely combinatorial nature. It is therefore easy to analyze the groups, and many combinatorial properties of the resulting groups can be tuned. In fact, we can start with an arbitrary graph and encode it into a group while preserving isomorphisms. We should stress that even though we start with graphs of unbounded Weisfeiler-Leman dimension, the resulting groups have only dimension 3. This highlights the power of Weisfeiler-Leman-type algorithms to distinguish groups beyond the scope of traditional invariants.

\subsection{Further related work}

Our work can be understood as studying the descriptive complexity of finite groups.  
We refer to Grohe's monograph~\cite{MR3729479} for extensive information on the descriptive complexity of graphs (rather than groups). A central result  in~\cite{MR3729479} shows graph classes with a forbidden minor have bounded WL-dimension. A recent paper relating first order logics and groups is~\cite{MR3705849}.
The descriptive complexity of finite abelian groups has been studied in~\cite{DBLP:journals/ijac/Gomaa10}. However, descriptive complexity of groups has been investigated considerably less than that of graphs. In contrast to this, the research body on the algorithmic Group Isomorphism Problem is extensive. The results can generally be divided into research with a more practical and research with a more theoretical focus.

\textit{On the practical side} the best algorithms for isomorphism testing are typically implemented in computer algebra systems such as SAGE, MAGMA, GAP. 
Classical algorithms include the one by Smith~\cite{smith_1996} (for solvable groups), the one by Eick, Leedham-Green and
O'Brien (for $p$-groups)~\cite{MR1904637,MR1283739},  as well as a general algorithm by Cannon and Holt~\cite{DBLP:journals/jsc/CannonH03}. 
Newer algorithms have been developed by Wilson~\cite{WilsonThesis} with numerous improvements over time together with Brooksbank and Maglione~\cite{MR3591162}.
More recent work introduces ever stronger invariants to distinguish groups quickly. We refer to~\cite{DBLP:journals/corr/abs-1905-02518} for an overview and the most recent techniques and an algorithm incorporating many of them.
Dietrich and Wilson report that current isomorphism tests are already infeasible in practice  on some groups with orders in the thousands~\cite{1806.08872}. 

In any case, in our work we focus on the \textit{theoretical side}.
As mentioned before, the best bound for the general problem is by Rosenbaum~\cite{DBLP:journals/corr/abs-1304-3935}. Polynomial time algorithms have been developed for various classes of groups~\cite{DBLP:conf/soda/BabaiCGQ11,DBLP:conf/icalp/BabaiCQ12,DBLP:conf/stacs/BabaiQ12,DBLP:conf/csr/DasS19,DBLP:conf/stacs/Gall09,DBLP:journals/siamcomp/GrochowQ17,DBLP:journals/siamcomp/IvanyosQ19, DBLP:journals/jcst/QiaoNT12}. There is an algorithm running in polynomial time for most 
orders~\cite{1806.08872}. For the currently fastest isomorphism algorithm for permutation groups 
see \cite{DBLP:conf/soda/Wiebking20}.

Recent efforts incorporate the Weisfeiler-Leman algorithm into the group isomorphism context~\cite{DBLP:journals/corr/abs-1905-02518,DBLP:conf/focs/LiQ17}.
However, there is a crucial difference 
to our work. Indeed, in these papers the authors use a combinatorial construction within the groups acting on vector spaces on which the (graph) WL-algorithm is executed. This is different to the general algorithm for all groups defined here.
Thus, a priori the two algorithmic approaches are unrelated, warranting further study.

\section{Preliminaries and notation}

\paragraph{Groups.} Groups will be denoted by capital Latin characters. For a group $G$ and elements $g,h\in G$ we write their \textit{commutator} as $[g,h]:=ghg^{-1}h^{-1}$ and we use $G'$ to refer to the subgroup of $G$ generated by all commutators. Then $G'$ is the unique minimal normal subgroup of $G$ with abelian quotient. The \textit{centralizer} of $x\in G$ is
$C_G(x):=\{g\in G\mid [x,g]=1\}$ and then $Z(G):=\{g\in G\mid C_G(g)=G\}$ is the \emph{center} of $G$. For a prime $p$, a group is called a $p$-group if $|G|=p^n$ is a power of $p$ (in particular, we assume $G$ to be finite here). The \emph{exponent} of a group is the least common multiple of the orders of its elements. A $p$-group $G$ is \textit{elementary abelian} if it is abelian and of prime exponent (i.e., $G\cong \mathbb{F}_p^n$ for some $n$). The \textit{Frattini-subgroup} $\Phi(G)$ of a group $G$ is the intersection of all maximal subgroups. If $G$ is a $p$-group then $\Phi(G)$ is the unique minimal normal subgroup of $G$ with elementary abelian quotient. The elements of $\Phi(G)$ are \emph{non-generators} in $G$, that is, if $\{g_1,\dots,g_m\}$ generates $G$ then so does $\{g_1,\dots,g_m\}\setminus\Phi(G)$.

We define a $1$-fold commutator to be just a regular commutator and then a $c$-fold commutator is an element of the form $[x,z]$ with $x\in G$ and $z$ a $(c-1)$-fold commutator in $G$.
A group $G$ is said to be \textit{nilpotent} if there is some integer $c$ such that $c$-fold commutators are always trivial in $G$ and if this is the case then the \textit{nilpotency class} of $G$ is the smallest such $c$. For example abelian groups are exactly the groups of nilpotency class $1$ and a group has nilpotency class $2$ if and only if it is non-abelian and every commutator is central.

A \textit{group isomorphism} is a bijective map $\varphi\colon G\to H$ that preserves group multiplication. We collect all isomorphisms between $G$ and $H$ in a (possibly empty) set $\Iso(G,H)$ and set $\Aut(G):=\Iso(G,G)$. We write $\Sub(G)$ for the set of all subgroups of $G$. 

We assume the term 'group' to mean 'finite group' and whenever we include infinite groups we do so explicitly.

\paragraph{Graphs.}
All graphs will be finite simple undirected graphs and referred to with greek characters, primarily $\Gamma$, subject to suitable subscripts. That is, a graph $\Gamma$ consists of a finite set of vertices $V(\Gamma)$ and a set of edges $E(\Gamma)\subseteq{{V(\Gamma)}\choose{2}}:=\{M\subseteq V(\Gamma)\mid |M|=2 \}$. The \emph{complement} of $\Gamma$ will always be the \textit{simple} complement graph, namely $\co(\Gamma):=\left(V(\Gamma),{{V}\choose{2}}-E(\Gamma)\right)$. 
The set of \emph{neighbors} of $v\in V(\Gamma)$ is $N(v):=\{w\in V(\Gamma)\mid \{v,w\}\in E\}$ and $N[v] := N(v)\cup\{v\}$ is the \emph{closed neighborhood} of $v$. 
The \textit{degree} of $v$ is $d(v):=|N(v)|$. A graph is \textit{$d$-regular} if every vertex has degree $d$. For a set of vertices $M\subseteq V(\Gamma)$, the \textit{induced subgraph} is $\Gamma[M]:=\left(M,E(\Gamma)\cap
{{M}\choose{2}}\right)$.

An \textit{isomorphism} of graphs is a bijective map $\varphi\colon V(\Gamma_1)\to V(\Gamma_2)$ that simultaneously preserves edges and non-edges. The set of isomorphisms between $\Gamma_1$ and $\Gamma_2$ is $\Iso(\Gamma_1,\Gamma_2)$ and define $\Aut(\Gamma_1):=\Iso(\Gamma_1,\Gamma_1)$.

The \emph{commuting graph} of a group is the graph whose vertices are the group elements and two distinct elements~$g,g'$ are adjacent if~$[g,g']=1$.

\subsection{The WL-algorithm for graphs}
Before we explore how the WL-algorithm can be applied to groups, we briefly recapitulate its classic definition for graphs. Given a graph~$\Gamma$, the~$k$-dimensional version of the algorithm for positive~$k\in \mathbb{N}$ repeatedly colors the~$k$-tuples of vertices with abstract colors that encode how each tuple is situated within the graph.
The initial coloring of each tuple~$(g_1,\ldots,g_k)$ encodes the isomorphism type of the graph induced by~$\{g_1,\ldots,g_k\}$, taking into account where the vertices occur in the tuple. Specifically, a coloring~$\chi_0\colon V(\Gamma)^k\rightarrow S$ into some set of colors is defined so that~$\chi_0(g_1,\ldots,g_k) = \chi_0(g'_1,\ldots,g'_k)$ holds exactly if there is an isomorphism from $\Gamma[\{g_1,\ldots,g_k\}]$ to~$\Gamma[\{g'_1,\ldots,g'_k\}]$  which sends~$g_i$ to~$g'_i$ for all~$i\in \{1,\ldots,k\}$.

The coloring is now iteratively refined as follows. For a tuple~$\bar{g}=(g_1,\ldots,g_k)\in V(\Gamma)^k$ and~$x\in V(\Gamma)$, define~$\bar{g}_{|i\leftarrow x}$ to be the tuple~$(g_1,\ldots,g_{i-1},x,g_{i+1},\ldots,g_k)$ obtained by replacing the~$i$-th entry with~$x$. Then we define for~$k>1$ the coloring~$\chi_i(\bar{g}):=$
\[
	\left( \chi_{i-1}(\bar{g}),\lmulset  (\chi_{i-1}(\bar{g}_{|1\leftarrow x}),\dots,\chi_{i-1}(\bar{g}_{|k\leftarrow x}))\mid x\in V(\Gamma)\rmulset\right).
\]

Here~$\lmulset \cdots \rmulset$ denotes multisets. 
Thus the color of the next iteration consists of the color of the previous iteration and the multiset of colors obtained by replacing each entry in the tuple with another vertex from the graph. For~$k=1$ the definition is slightly different, namely that for~$\chi_i(\bar{v})$ the multiset is only taken over vertices~$x$ in the neighborhood~$N(v)$.

Adding the color of the previous iteration as first entry ensures that the partition induced on~$V(\Gamma)$ by~$\chi_i$ is finer than (or as fine as) the partition induced by~$\chi_{i-1}$. Let~$j$ be the least positive integer for which the partition induced by~$\chi_{j-1}$ agrees with the partition induced by~$\chi_j$, then we define the final coloring~$\chi_{\infty}$ to be~$\chi_{j-1}$.
Since the domain of the~$\chi_i$ has size~$|V(\Gamma)|^k$, we know that~$j\leq |V(\Gamma)|^k$. For fixed~$k\in \mathbb{N}$, it is possible to compute the partition of~$\chi_{\infty}$ in time~$\mathcal{O}(n^{k+1}\log (n))$~\cite{MR1060782}.

To distinguish two non-isomorphic graphs the algorithm is applied on the disjoint union. If in the final coloring the multiset of colors appearing in one graph is different than those appearing in the other graph, then the graphs are not isomorphic. The converse does not necessarily hold, as we explain next.

\subsection{The CFI-graphs}\label{sec:CFI-graphs}

As mentioned previously, for each~$k$ there is a pair of non-isomorphic graphs not distinguished by~$k$-WL.

\begin{thm}[Cai, F\"urer, Immerman~\cite{DBLP:journals/combinatorica/CaiFI92}]\label{thm:cfi}
There is an infinite family of pairs of non-isomorphic 3-regular graphs on~$\mathcal{O}(k)$ vertices not distinguished by the~$k$-dimensional Weisfeiler-Leman algorithm. 
\end{thm}

\tikzstyle{normalvertex}=[circle, draw]
\begin{figure}
\centering
\begin{tikzpicture}[thick,scale =0.6]
  \node[style=normalvertex,fill=red,label=left:{$b_1$}] (a0) at (-1.5,5) {};
  \node[style=normalvertex,fill=red,label=left:{$a_1$}] (b0) at (-3.5,5) {};
  
  \node[style=normalvertex,fill=blue,label=left:{$b_3$}] (a2) at (1,9) {};
  \node[style=normalvertex,fill=blue,label=left:{$a_3$}] (b2) at (-1,9) {};
  
  \node[style=normalvertex,fill=green,label=left:{$b_2$}] (a1) at (4,5) {};
  \node[style=normalvertex,fill=green,label=left:{$a_2$}] (b1) at (2,5) {};
  
  \node[style=normalvertex,label=180:{{\small $000$}}] (v1) at (-3,7) {};
  \node[style=normalvertex,label=180:{{\small $011$}}] (v2) at (-1,7) {};
  \node[style=normalvertex,label=180:{{\small $110$}}] (v3) at (1,7) {};
  \node[style=normalvertex,label=180:{{\small $101$}}] (v4) at (3,7) {};
  
  \path[thick]
   (v1) edge (b0)
   (v1) edge (b2)
   (v1) edge (b1)
   (v2) edge (b0)
   (v2) edge (a2)
   (v2) edge (a1)
   (v3) edge (a0)
   (v3) edge (b2)
   (v3) edge (a1)
   (v4) edge (a0)
   (v4) edge (a2)
   (v4) edge (b1);

 \end{tikzpicture}\caption{A depiction of the CFI-gadget~$F_3$.}\label{fig:cfi}
\end{figure}
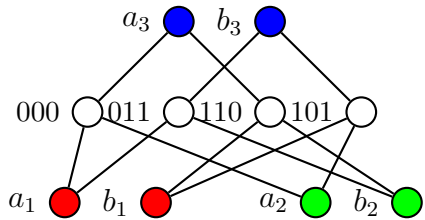

Since we intend to exploit the construction by transferring it to groups, we describe it next. We start with a connected base graph~$\Gamma$.
In this graph every vertex is replaced by a particular gadget and the gadgets are interconnected according to the edges of~$\Gamma$ as follows.
For a vertex~$v$ of degree~$d$ we use the gadget~$F_d$, which is a graph whose vertex set consists of external vertices~$O_d= \{a^v_1,b^v_1,a^v_2,b^v_2,\ldots,a^v_d,b^v_d\}$ and internal vertices $M_d$. The internal vertices form a copy of the set of those 0-1-strings of length $d$ that have an even number of entries equal to 1. For each~$i$, each internal vertex~$m$ is adjacent to exactly one vertex of~$\{a^v_i,b^v_i\}$, namely it is adjacent to~$a_i$ if the~$i$-th bit of the string~$m$ is 0 and to~$b_i$ otherwise. An example of~$F_3$ is depicted in Figure~\ref{fig:cfi}.
It remains to explain how the different gadgets are interconnected. For this, for a vertex~$v\in\Gamma$ of degree~$d$ each edge is associated with one of the pairs~$a^v_i,b^v_i$. For an edge~$(u,v)\in E(\Gamma)$ assume~$u$ is associated with the pair~$(a^u_i,b^u_i)$ in the gadget corresponding to~$u$ and~$v$ is associated with the pair~$(a^v_j,b^v_j)$ in the gadget corresponding to~$v$. Then we insert (parallel) edges~$\{a^u_i,a^v_j\}$ and~$\{b^u_i,b^v_j\}$. Adding such parallel edges for each edge of the base graph we obtain the graph~$\CFI(\Gamma)$. The \emph{twisted CFI-graph} $\widetilde{\CFI(\Gamma)}$ is obtained by replacing one pair of (parallel) edges~$\{a^u_i,a^v_j\}$ and~$\{b^u_i,b^v_j\}$ with the (twisted) edges~$\{a^u_i,b^v_j\}$ and~$\{b^u_i,a^v_j\}$. It can be shown that for connected base graphs (up to isomorphism) it is irrelevant which edge is twisted~\cite{DBLP:journals/combinatorica/CaiFI92}. For a subset of the edges of the base graph~$E'\subseteq E(\Gamma)$, we can define the graph obtained by twisting exactly the edges in~$E'$. The resulting graph is isomorphic to~$\CFI(\Gamma)$ if~$|E'|$ is even and isomorphic to $\widetilde{\CFI(\Gamma)}$  otherwise.

In the original construction the base graph is usually thought of as vertex colored with all vertices obtaining a different color. This makes all gadgets distinguishable. The colors can be removed by attaching gadgets retaining the property that the base graph is identified by 2-dimensional Weisfeiler-Leman. We want to record here the observation that it is possible to choose the base graph of WL-dimension 2 while maintaining the property that it is 3-regular.

\begin{obs}\label{obs:base:graph:low:WL:dim}
The 3-regular base graph~$\Gamma$ can be chosen to have Weisfeiler-Leman dimension at most 2.
\end{obs}

This can be seen in two ways, by adding gadgets on edges or by observing that random expanders, usually used in the construction, have this property.

\subsection{First order logic with counting}\label{subsec:first:order:logic:graphs}

There is a close connection between the Weisfeiler-Leman algorithm of dimension~$k$ and the~$(k+1)$-variable fragment of first order logic on graphs with counting quantifiers~\cite{DBLP:journals/combinatorica/CaiFI92}. To obtain this logic we endow first order logic with counting quantifiers. The formula~$\exists^{\geq i}x  \varphi(x)$ expresses then the fact that there are at least~$i$ distinct elements that satisfy the formula~$\varphi$. For example the formula~$\exists^{\geq 3} x \exists^{\geq 4} y E(x,y)$ would express that the graph contains at least 3 vertices of degree at least 4. 
The logic~$\mathbf{C}^{k}$ is the fragment of said logic which allows formulas to only use~$k$ distinct variables (that can however be reused an arbitrary number of times). We refer to~\cite{MR1060782} for a more thorough introduction to these logics and a proof that two graphs can be distinguished by~$k$-dimensional WL exactly if there is a formula in~$\mathbf{C}^{k+1}$ that holds on the one graph but not on the other.
Often such logics are endowed with a fixed-point operator, but since we will only apply the formulas to structures of fixed size, this will not be necessary for us (see~\cite{DBLP:books/cu/O2017} for more information).

\subsection{The pebble game}\label{subsec:pebble:game:graphs}

There is a third concept, the bijective pebble game~\cite{DBLP:journals/iandc/Hella96}, that has a deeper connection to the logic~$\mathbf{C}^{k+1}$ and the~$k$-WL. This game is often used to show that graphs cannot be distinguished by $k$-WL.
The game is an Ehrenfeucht-Fra{\"{\i}}ss{\'{e}}-type game with two players Duplicator and Spoiler. Initially~$k+1$ pairs of pebbles, each pair uniquely colored, are placed next to two given input graphs~$\Gamma_1,\Gamma_2$.
Each round proceeds as follows: Spoiler picks up a pebble pair~$(p_i; p_i')$ of pebbles of the same color.
Then Duplicator chooses a bijection~$\varphi$ from~$V(\Gamma_1)$ to~$V(\Gamma_2)$. Then Spoiler places pebble~$p_i$ on a vertex~$v\in V(\Gamma_1)$ and places~$p_i'$ on~$\varphi(v)$. Spoiler wins if at any point in time the graph induced by the vertices occupied by pebbles
in~$V(\Gamma_1)$ is not isomorphic to the graph induced by the vertices occupied by
pebbles in~$V(\Gamma_2)$ via a map that sends a pebble~$p_i$ to its corresponding pebble of the same color~$p_i'$ in the other graph. Spoiler also wins (in round 0) if~$|V(\Gamma_1)|\neq |V(\Gamma_2)|$.

When using~$k+1$ pebbles on two graphs, the game can be won by Spoiler exactly if~$k$-WL distinguishes the graphs~\cite{DBLP:journals/iandc/Hella96}.

\section{WL-type algorithms on groups}

As with graphs we would like to be able to study combinatorial properties of finite groups using WL-type algorithms. The natural approach is to adapt the methods from the last section to suit (finite) groups. However, depending on the interpretation of these methods, we will obtain several different choices for initial colorings and refinement strategies for finite groups. We will argue that different methods are all in some sense natural and interesting in their own right. However, the different concepts (possibly) lead to different notions of Weisfeiler-Leman dimension for groups. In contrast to this, for graphs, all notions are equivalent.
While we are not able to precisely determine whether for groups the different methods are equally powerful at this point, we do however show that exchanging one method for another changes the dimension by at most a constant factor.

\subsection{Weisfeiler-Leman algorithms for groups}

The following algorithms define color-refinement procedures on $k$-tuples of group elements. Since groups are (essentially) ternary relational structures, we will usually require for the dimension that~$k\geq 2$. In the following let~$G$ be a group. We will define three versions of the WL-algorithm for groups.

\paragraph{Version I:} Define an initial coloring~$\chi_0\colon G^k \rightarrow C$ on $k$-tuples of group elements so that $(g_1,\dots,g_k)$ and $(h_1,\dots,h_k)$ obtain the same color if and only  if for all indices~$i,j\in \{1,\ldots,k\}$ we have~$g_i=g_j$ exactly if~$h_i=h_j$ and for all indices $i,j,m\in \{1,\ldots,k\}$ we have $g_ig_j=g_m$ exactly if $h_ih_j=h_m$. We iteratively define the refinement~$\chi_i$ in the classical way just like it is defined for graphs, that is for $\bar{g}:=(g_1,\dots,g_k)\in G^k$ we have~$\chi_i(\bar{g}):=$
\[
	\left( \chi_{i-1}(\bar{g}),\lmulset  (\chi_{i-1}(\bar{g}_{|1\leftarrow x}),\dots,\chi_{i-1}(\bar{g}_{|k\leftarrow x}))\mid x\in G\rmulset\right).
\]

\paragraph{Version II:} In the definition for graphs, the initial coloring of a tuple takes into account the subgraph induced by the tuple. In analogy to this, one might argue that for groups the initial coloring needs to take into account the subgroup generated by the tuple.
Thus, in Version II, we define an initial coloring~$\chi_0$ on $k$-tuples of group elements such that $(g_1,\dots,g_k)$ and $(h_1,\dots,h_k)$ obtain the same color if and only if there is a map with $g_i\mapsto h_i$ which extends to an isomorphism from~$\langle g_1,\ldots, g_k\rangle $ to~$\langle h_1,\ldots, h_k \rangle$. In this case we will also say the tuples agree in their \textit{marked isomorphism type}.
The iterative refinement is again performed in the classical way.

\paragraph{Version III:} 

For Version III, we encode groups as graphs, execute the WL-algorithm for graphs and pull back the coloring.
For this choose an isomorphism-preserving, invertible functor $\Gamma_{\bullet}$ that maps finite groups $G$ to finite (simple) graphs $\Gamma_G$. Our working example will be the following construction but other choices are certainly possible and many choices will lead to equivalent or at least related algorithms.

To obtain the graph $\Gamma_G$ (see Figure~\ref{fig:mult:gadet}), start with a set of isolated nodes corresponding to elements of the group $G$. For each pair of group elements $(g,h)$ add a multiplication gadget $M(g,h)$ by adding 4 nodes~$a_{gh},b_{gh},c_{gh},d_{gh}$ and add the edges~$E(M(g,h))=$ \[ \{\{g,\!a_{gh}\},\{h_{gh},\!b_{gh}\},\{gh,\!d_{gh}\},\{a_{gh},\!b_{gh}\},\{b_{gh},\!c_{gh}\},\{c_{gh},\!d_{gh}\}\}.\]

We then use the classical $k$-dimensional WL-algorithm on the graph $\Gamma_G$ and pull back the colorings of $k$-tuples by simply restricting it to $G^k$.

By construction we have $|\Gamma_G|=\Theta(|G|^2)$ and due to vertex-degrees $G$ is a canonical subset of the vertices of $\Gamma_G$. Thus for two groups~$G,H$ we have~$\Gamma_G\cong \Gamma_H$ if and only if~$G\cong H$. 

Many other reductions of this form transforming groups to graphs are possible. However, some of them are artificial. For example, one could artificially ensure that the resulting graph has low WL-dimension by precomputing certain isomorphism-invariants not captured by the WL-algorithm. Thus, a unified treatment of all isomorphism-preserving constructions seems infeasible. 
However, it seems that for many 'well-behaved' functors the WL-dimension of the constructed graphs differ	 by a constant factor only. 
On another note, it would be interesting to obtain efficiently computable subquadratic reductions from groups to graphs, but we are not aware of such a construction.

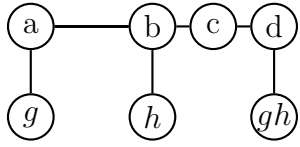
\begin{figure}
\centering
\begin{tikzpicture}[thick,scale=0.8]
			\tikzstyle{every node}=[draw,shape=circle, inner sep=0, minimum height= 0.6cm]
			\path 
			(-7.5,-1.5) node (p1) {$g$}
			(-7.5, 0) node (p2) {a}
			(-5.5,-1.5) node (p3) {$h$}
			(-5.5,0) node (p4) {b}
			(-4.5,0) node (p5) {c}
			(-3.5,0) node (p6) {d}
			(-3.5,-1.5) node (p7) {$gh$};
			\draw[thick]
			(p1) -- (p2)
			(p2) -- (p4)
			(p3) -- (p4)
			(p2) -- (p4)
			(p5) -- (p4)
			(p5) -- (p6)
			(p6) -- (p7);
	\end{tikzpicture}
\caption{The multiplication gadget to encode the multiplication~$g\cdot h = gh$}\label{fig:mult:gadet}
\end{figure}

For fixed~$k$, each version of the WL-algorithm gives rise to a polynomial-time (possibly) partial isomorphism test on pairs of finite groups. Indeed, marked isomorphism of $k$-tuples in a group $G$ can be checked in time $\mathcal{O}(|G|\log(|G|))$ so we can compute initial colors in time $\mathcal{O}(|G|^{k+1}\log(|G|))$ for Versions I and II.
The refinement steps are the same as for graphs and thus we obtain the same $\mathcal{O}(|G|^{k+1}\log(|G|))$ bound for the rest of the computation. For Version III we have a quadratic blowup yielding time~$\mathcal{O}(|G|^{2k+1}\log(|G|))$.

\begin{dfn}
	Groups $G$ and $H$ are equivalent with respect to $k$-WL in Version $J\in\{\text{\versone,\verstwo,\versthree}\}$, in symbols $G\equiv_{WL_k^J} H$, if there is a bijection $f\colon G^k\to H^k$ preserving final colors of the respective color-refinement procedure. Furthermore we write $WL_k^J\preceq WL_{k'}^{J'}$ if it holds that
	$G\equiv_{WL_{k'}^{J'}} H\Rightarrow G\equiv_{WL_k^J} H$, i.e., the distinguishing power of $WL_k^J$ is weaker than or equal to the distinguishing power of $WL_{k'}^{J'}$.
\end{dfn}

The main result of this section is that we can exchange one version for another when we multiply the dimension with a constant factor. 
When studying WL-type algorithms it is often useful to have equivalent pebble games at hand, so we first associate a pebble game to each of the variants above.

\subsection[Bijective k-pebble games]{Bijective $k$-pebble games} We now define suitable pebble games for the different versions. Each of these games is played by two players Spoiler and Duplicator and in each case we will say that Duplicator wins the game if and only if there is a strategy for Duplicator to keep the game going on forever. The board consists of a pair of finite groups $G,H$ of equal order (or rather their elements) or a pair of corresponding graphs $\Gamma_G,\Gamma_H$ for Version III, respectively. There are~$k$ pairs of pebbles~$(p_1;p'_1), \ldots (p_k;p'_k)$. We think of pebbles in the same pair as having the same color, and pebbles from different pairs as having distinct colors. The pebbles can be placed beside the board or on the group elements (graph vertices in Version III), in which case we say a group element is pebbled. Pebbles~$p_i$ are placed on elements of~$G$ (vertices of~$\Gamma_G$) and pebbles~$p'_i$ on~$H$ ($\Gamma_H$). At any point in time the pebbles~$(p_1,\ldots,p_k)$ give us a pebbled tuple in~$(G\cup \{\bot\})^k$ (or $(\Gamma_G\cup\{\bot\})^k$), where~$\bot$ indicates that the pebble is placed besides the board.

\paragraph{Version I:} All $k$ pairs of pebbles are initially placed beside the board. A round of the game consists of these steps:
\begin{enumerate}
	\item Spoiler picks up a pair of pebbles~$(p_i;p_i')$.
	\item Duplicator chooses a bijection $f\colon G\to H$.
	\item Spoiler pebbles some element $g\in G$ with~$p_i$, the corresponding pebble~$p_i'$ is placed on $f(g)$.
\end{enumerate}

The winning condition is always checked right after Step 1. At that moment, the pebbles not in Spoiler's hand then pebble a~$k$-tuple over $G\cup\{\bot\}$ and a corresponding $k$-tuple over $H\cup\{\bot\}$. Spoiler wins if the pebbled tuples differ with respect to the initial coloring of Version I. (This implies that no more pebbles are placed beside the graph.)

\paragraph{Version II:} Version II differs from Version I only in that the winning condition uses the initial coloring of Version II rather than Version I. That is, Spoiler wins if the map induced by pairs of pebbles does not extend to an isomorphism between the subgroups generated by the pebbled group elements, and the game continues otherwise.

\paragraph{Version III:} Version III is the (classical) bijective $k$-pebble game for graphs played on $\Gamma_G$ and $\Gamma_H$ (see Subsection~\ref{subsec:pebble:game:graphs}).

\bigskip

In the pebble games, when we say that ``Duplicator has to do something'', we mean that otherwise Spoiler wins the game.
We say that Duplicator respects a certain property of group elements if Duplicator always has to pebble pairs of groups elements which agree in whether they have the property.
One can show that Duplicator must respect the partial mapping given by the pairs of pebbles that are currently on the board. Indeed, otherwise Spoiler can win in the next round by pebbling the location where this is violated.

For each game we can also use initial configurations of pebbled tuples instead of starting from empty configurations.

\textit{Remark:} 
	Color refinement and pebble games are not necessarily restricted to finite groups. While not clear that the results are computable, they still may be of theoretical interest. The same goes for the logics defined next.

\subsection{Logics with counting} As for graphs, the $k$-dimensional refinement on groups can also be interpreted in terms of first-order counting logic.

Recall the central aspects of first order logic. There is a countable set of variables~$\{x_1,x_2,\ldots\}$. Formulas are inductively defined so that~$x_i=x_j$ is a formula for all pairs of variables and if~$\varphi$ is a formula then~$\varphi\wedge \varphi$,~$\varphi\vee \varphi$,$\neg \varphi$,~$\exists x_i \varphi$ and~$\forall x_i \varphi$ are formulas. The semantics are defined in the obvious way. 
First order logic with counting allows additionally formulas of the form $\exists^{\geq t}x_i \varphi(x_i)$ with the semantic meaning that there are at least~$t$ distinct elements that satisfy~$\varphi$.

To define logics on groups we need to additionally define a relation that relates to the group multiplication.

\paragraph{Version I:} In Version I we add a ternary relation~$R$ with which we can create terms of the form~$R(x_i,x_j,x_\ell)$. The semantic interpretation is that~$R(x_i,x_j,x_\ell)$ holds if~$x_i \cdot x_j = x_{\ell}$.
We call~$\mathcal{L}_\versone$ the first order logic with counting on groups arising this way and let~$\mathcal{L}^k_\versone$ be its~$k$-variable fragment.

\paragraph{Version II:} For~$\mathcal{L}_{\verstwo}$ we use a different relation to access multiplication: 
The relation~$R(x_{i_1},x_{i_2},\ldots,x_{i_t};w)$ holds, where~$w\in (\{x_{i_1},x_{i_2},\ldots,x_{i_t}\}\cup \{x_{i_1}^{-1},x_{i_2}^{-1},\ldots,x_{i_t}^{-1}\})^*$ is a word in the~$x_{i_j}$, if multiplying the elements according to~$w$ gives the trivial element. For example in an abelian group~$G$ the relation~$R(a,b;aba^{-1}b^{-1})$ would hold for all elements~$a,b\in G$. The relation~$R(a;aa)$ would only hold if~$a$ is the trivial element. We let~$\mathcal{L}^k_{\verstwo}$ be the~$k$-variable fragment of the logic.
Note that for~$\mathcal{L}^k_{\verstwo}$ it actually suffices to use 
only~$k+1$ entries in the relation.

\paragraph{Version III:} The natural choice of logic for Version III is of course the classical first order logic with counting~$\mathbf{C}$ on graphs as discussed in the preliminaries (Subsection~\ref{subsec:first:order:logic:graphs}), where we have the relation~$E(u,v)$ to encode edges. For notational consistency we define $\mathcal{L}^k_{\versthree} \coloneqq \mathbf{C}^{k}$ to be the~$k$-variable fragment of this logic.

\subsection{Equivalence between the different concepts}\label{subsec:equivalence:of:concepts}

For each of the versions we have defined, we sketch the arguments for equivalence of the expressive power between the WL-algorithm, the pebble game, and the corresponding logic.
Let us fix groups $G$ and $H$ of the same order. The argument basically follows other well known arguments to show such equivalences (see e.g., \cite{DBLP:journals/combinatorica/CaiFI92}).

\begin{thm}\label{thm:game:and:algo:agree}
	Two groups $G$ and $H$ are distinguished by the $k$-WL-refinement (Version $J\in \{\text{\versone,\verstwo,\versthree}\}$) if and only if the same holds for the bijective $k+1$-pebble game (Version $J$).
\end{thm}

\textit{Remark:} Let us remark on a small detail where the group situation can differ from that of graphs. Note that in our definition of the game, the winning condition is only ever checked after Step 1. We could also check the winning condition when~$k+1$ pebble pairs are situated on the group after a round is finished. For this we would need an initial coloring that works with~$k+1$ tuples. For graphs this change does not make a difference, since the winning condition only ever depends on 2 pebble pairs. Similarly for Version I, where the winning condition occurs due to 3 pebble pairs, if~$k>3$ then it is irrelevant when we check the winning condition. However, for Version II we are not so sure how the power of the game changes, when altering the moment at which the winning condition is checked.

\begin{thm}\label{thm:game:and:logic:agree}
	$G\not\equiv_{WL_k^J}H$ if and only if there is a sentence in~$\mathcal{L}^{k+1}_J$ that holds on one of the groups but not the other.
\end{thm}

The rest of this section contains the proof of Theorems~\ref{thm:game:and:algo:agree} and~\ref{thm:game:and:logic:agree}.

\begin{lem}
	Suppose $\bar{g} := (g_1,\dots,g_k)\in G^k$ and $\bar{h} := (h_1,\dots,h_k)\in H^k$. If $\bar{g}$ and $\bar{h}$ obtain different colors 
	in the $i$-th iteration of $k$-dimensional WL-refinement then Spoiler can win the $(k+1)$-pebble game in $i$ moves on initial configuration $(\bar{g},\bar{h})$.
	(Here we use the same version for WL-refinement and pebble game.)
\end{lem}
\begin{proof}[Proof Sketch]
	\emph{(Version I.)} For $i=0$ there is nothing to show. Assume now that $i>0$. By assumption $\chi_i(\bar{g})$ and $\chi_i(\bar{h})$ are different which means that either we already have
	$\chi_{i-1}(\bar{g})\neq\chi_{i-1}(\bar{h})$ or there is no color-preserving matching between the
	tuples $(\chi_{i-1}(\bar{g}_{|1\leftarrow x}),\dots,\chi_{i-1}(\bar{g}_{|k\leftarrow x}))$ for $x\in G$ and tuples $(\chi_{i-1}(\bar{h}_{|1\leftarrow y}),\dots,\chi_{i-1}(\bar{h}_{|k\leftarrow y}))$ for $y\in H$. In other words, no matter which bijection $f\colon G\to H$ with $f(g_j)=h_j$ Duplicator chooses there will be some $x\in G$ and some position $1\leq j\leq k$
	such that $\chi_{i-1}(\bar{g}_{|j\leftarrow x})\neq \chi_{i-1}(\bar{h}_{|j\leftarrow f(x)})$ and Spoiler can make progress by changing the $j$-th pebble from $g_j$ to $x$. By induction, Spoiler now has a winning strategy with $i-1$ moves while having moved at most once.
	
	\emph{(Version II.)}
	If $i=0$ then $\bar{g}$ and $\bar{h}$ differ with respect to marked isomorphism, thus Spoiler can win without moving at all.
	For $i>0$ the argument is the same as before since the refinement steps are defined equally. 
	
	\emph{(Version III.)} This is exactly the classical result  for graphs, see~\cite{DBLP:journals/iandc/Hella96,DBLP:journals/combinatorica/CaiFI92}.
\end{proof}

\begin{lem}
	Suppose $\bar{g} := (g_1,\dots,g_k)\in G^k$ and $\bar{h} := (h_1,\dots,h_k)\in H^k$. If Spoiler can win the $(k+1)$-pebble game in $i$ moves on initial configuration $(\bar{g},\bar{h})$
	then $\bar{g}$ and $\bar{h}$ obtain different colors in the $i$-th iteration of $k$-dimensional WL-refinement. (Again we use the same version for WL-refinement and pebble game.)
\end{lem}
\begin{proof}[Proof sketch]
	\emph{(Version I.)} If $i=0$ then the initial configuration is already a winning one for Spoiler which is by definition the same as
	$k$-tuples getting different initial colors. By induction, for any bijection $f\colon G\to H$ Duplicator may choose, Spoiler can reach in one move a configuration $(\bar{g_1},\bar{h_1})$ where 
	$\bar{g_1}=\bar{g}_{|j\leftarrow x}$ and $\bar{h_1}=\bar{h}_{|j\leftarrow f(x)}$ for some position $j$ and such that $\chi_{i-1}(\bar{g_1})\neq \chi_{i-1}(\bar{h_1})$.
	Since this is true for any possible bijection, the tuples $\bar{g}$ and $\bar{h}$ already have to differ with respect to $\chi_i$.
	
	\emph{(Version II.)} The argument is the same as for Version I.

	\emph{(Version III.)} This is again a classical result~\cite{DBLP:journals/iandc/Hella96,DBLP:journals/combinatorica/CaiFI92}.
\end{proof}

\begin{proof}[Proof of Theorem~\ref{thm:game:and:algo:agree}]
Theorem~\ref{thm:game:and:algo:agree} follows immediately from the previous two lemmas.
\end{proof}

It remains to argue the equivalences between the logic and the pebble game for each of the versions. This again basically follows from known techniques.

We first argue that for Version I and~II the quantifier free formulas of the~$k$-variable fragment of each version characterize the initial colorings.

\begin{lem}\label{lem:quantifier-free:formula:is:inital:col}
	There is a quantifier free $k$-variable formula $\varphi(x_1,\dots,x_k)\in\mathcal{L}_J$ distinguishing $k$-tuples $\bar{g}$ and $\bar{h}$ if and only if these tuples differ in their initial coloring in version $J\in \{\text{\versone,\verstwo}\}$.
\end{lem}
\begin{proof}
	If $J=1$ then $\varphi$ distinguishes the tuples $\bar{g}$ and $\bar{h}$ if and only if there is an atomic statement of the form $x_i=x_s$ or $R(x_i,x_j,x_s)$, interpreted as $x_i\cdot x_j = x_s$, with respect to which $\bar{g}$ and $\bar{h}$ differ. This is precisely the definition of the Version I initial coloring. For Version $J=\verstwo$, if some word over $\bar{g}$ is (non)trivial but the corresponding word over $\bar{h}$ is not then clearly mapping $\bar{g}$ to $\bar{h}$ does not extend to an isomorphism. Assume now that $\bar{g}$ and $\bar{h}$ have different marked isomorphism types and w.l.o.g.~we have $ |\langle \bar{g}\rangle|\leq |\langle \bar{h}\rangle|$. By assumption every injective map between these groups extending $\bar{g}\mapsto\bar{h}$ is not multiplicative. This fact can be expressed in terms of a suitable word over $k$ symbols separating $\bar{g}$ from $\bar{h}$.
\end{proof}

\begin{lem}
	Suppose $\bar{g} := (g_1,\dots,g_k)\in G^k$ and $\bar{h} := (h_1,\dots,h_k)\in H^k$. For each Version~$J\in \{\text{I,II,III}\}$, the tuples $\bar{g}$ and $\bar{h}$ are distinguished by $k$-WL if and only if there is a formula~$\varphi(x_1,\ldots,x_k)$ in~$\mathcal{L}^{k+1}_J$
	such that~$\varphi(\bar{g})\nLeftrightarrow\varphi(\bar{h})$.
\end{lem}
\begin{proof}[Proof sketch]
\emph{(Version I.)} For Version I, due to Lemma~\ref{lem:quantifier-free:formula:is:inital:col}, the expressive power of the initial coloring is precisely the expressive power of quantifier-free formulas. The distinguishing power of the WL-algorithm on groups is thus equal to distinguishing power of the classical algorithm executed on a structure that is already endowed with the initial coloring. 
The equivalence between the $(k+1)$-variable fragment of the logic and the~$k$-WL algorithm for Version I on groups thus follows from the respective equivalence for graphs shown in~\cite{DBLP:journals/combinatorica/CaiFI92}.

\emph{(Version II.)} For Version II the argument is the same as for Version I except that we add the following observation: Since~$G$ and~$H$ are finite groups there is only a finite number of nonequivalent quantifier free formulas over~$\mathcal{L}_2$. By Lemma~\ref{lem:quantifier-free:formula:is:inital:col} tuples can be distinguished exactly if they obtain different colors in the initial coloring of Version II.

\emph{(Version III.)} This is again a classical result~\cite{DBLP:journals/combinatorica/CaiFI92}.
\end{proof}

\begin{proof}[Proof of Theorem~\ref{thm:game:and:logic:agree}]
Theorem~\ref{thm:game:and:logic:agree} follows immediately from the previous two lemmas.
\end{proof}

\subsection{Relationship between the different WL-algorithm versions}\label{subsec:relationship:between:versions}

Next, we want to relate different versions to each other and we will do so by exploiting the equivalence to pebble games.

\begin{dfn}
	Consider the $k$-pebble game on graphs $\Gamma_G$ and $\Gamma_H$ and assume that a pair of pebbles is placed on vertices corresponding to multiplication gadgets
	$M(g_1,g_2)$ and $M(h_1,h_2)$ (but not on vertices corresponding to group elements). Then the pairs $(g_1,h_1)$ and $(g_2,h_2)$ will be called \textit{implicitly pebbled}. Note that implicit pebbles always induce a pairing of group elements.
\end{dfn}

Intuitively, pebbling a vertex in a multiplication gadget is as strong as pebbling two group elements simultaneously, hence the definition of implicit pebbles. It can be shown that Duplicator has to respect the gadget structure and in particular the multiplication structure of the implicitly pebbled elements.

\begin{thm}\label{thm:relationship:between:versions}
	For all $k\in\mathbb{N}$ we have $WL_k^\versone\preceq WL_k^\verstwo\preceq WL_{k/2+2}^\versthree\preceq WL_{k+5}^\versone$.
\end{thm}

The rest of this section spans the proof of this theorem.

\begin{lem}\label{graph pebble game lemma}
	Consider the $k$-pebble game on graphs $\Gamma_G$ and $\Gamma_H$. If~$k\geq 4$ and one of the following happens
	\begin{enumerate}
		\item Duplicator chooses a bijection $f\colon \Gamma_G\to\Gamma_H$ with~$f(G)\neq H$,
		\item after choosing a bijection, there is a pebble pair~$(p,p')$, for which pebble~$p$ is on some vertex of $M(g_1,g_2)$ (not on $g_1$, $g_2$ or $g_1g_2$) and~$p'$ is on some vertex of $M(h_1,h_2)$ but $(f(g_1),f(g_2),f(g_1g_2))\neq (h_1,h_2,h_1h_2)$, or
		\item the map induced on group elements pebbled or implicitly pebbled by~$k-2$ pebbles does not extend to a group isomorphism between the corresponding generated subgroups
	\end{enumerate}
	then Spoiler can win the game. 
\end{lem}
\begin{proof}
	\begin{enumerate}
		\item In $\Gamma_G$ vertices corresponding to group elements have 	degree $3|G|$ while other vertices have degree $2$ or $3$.
		\item From now on always assume $f(G)=H$. Write $g_3:=g_1g_2$ and $h_3:=h_1h_2$. Let $f(g_i)\neq h_i$ for some $1\leq i\leq 3$ and put a pebble~$q$ other than~$p$ on $g_i$ and the corresponding one~$q'$ on $f(g_i)$. 
		If pebble~$p$ is not on the same type of vertex~(i.e., Type $a$,~$b$,~$c$ or~$d$, see Figure~\ref{fig:mult:gadet}) as pebble~$p'$ then Spoiler wins, since either the vertices have different degrees or their neighbors have different degrees.
		
		Now the pebbled vertex in
		$M(g_1,g_2)$ is connected to $g_i$ directly or via a path of non-group element vertices and for $f(g_i)$ and $M(h_1,h_2)$ either this is not the case or the path uses different types of vertices.
		Using a third pebble pair 
		Spoiler can explore this path and win on a configuration of three pebbles.
		\item By assumption there are at most $m:=2(k-2)$ implicitely pebbled pairs of group elements corresponding to at most $k-2$ pebbles that are currently on the board. We may assume that there are exactly $m$ such pairs, $(g_1,h_1),\dots,(g_m,h_m)$ say. By the second part of this lemma Duplicator has to choose some bijection $f$ such that $\bar{g}:=(g_1,\dots,g_m)\mapsto \bar{h}:=(h_1,\dots,h_m)$ respecting the pairing induced by indirectly pebbled group elements. By assumption this correspondence does not extend to an isomorphism between $\langle\bar{g}\rangle$ and $\langle\bar{h}\rangle$ so there must be a smallest word $w:=g_{i_1}\dots g_{i_t}$ over $\bar{g}$ such that
		\[
			f(w)\neq f(g_{i_1})\dots f(g_{i_t}).
		\] Spoiler can use an additional pair of pebbles to fix this image of $w$. Now Duplicator chooses a new bijection $f'$ on the remaining group elements. There is now either a smaller word over $\bar{g}$ with this property, in which case Spoiler moves the last pebble pair we just introduced to this new word and its image, or $w$ is still minimal with this property. The first case can only occur finitely many times and if $|w|=2$ Spoiler wins by part $2$. Thus assume that $w$ is still minimal. But then 
		\[
			f'(g_{i_1})f'(g_{i_2}\dots g_{i_t})=f'(g_{i_1})f'(g_{i_2})\dots f'(g_{i_t})\neq f'(w)
		\]and using a second additional pair of pebbles (now a total of at most $k$ pairs of actual pebbles) Spoiler can also fix the image of $g_{i_2}\dots g_{i_t}$ to be $f'(g_{i_2}\dots g_{i_t})$ and clearly wins from this configuration in at most two further rounds.
	\end{enumerate}
\end{proof}

\begin{lem}
	If $G\not\equiv_{WL_k^\verstwo} H$ then $G\not\equiv_{WL_{\lceil k/2\rceil +2}^\versthree} H$.
\end{lem}
\begin{proof}
Assume that Spoiler wins the $(k+1)$-pebble game in Version II. The idea is to simultaneously play a Version II game on groups and a Version III game on graphs. For this purpose we have to be able to compare pebble-configurations from the different games. Let $(\bar{g},\bar{h})$ be a configuration of pebbled $r$-tuples on $G$ and $H$ and call a configuration on graphs \textit{admissible} if it looks as follows: There is one pair of pebbles on the multiplication gadgets $M(g_1,g_2)$ and $M(h_1,h_2)$, another pair on  $M(g_3,g_4)$ and $M(h_3,h_4)$ and so forth. If $r$ is odd then there is another pair of pebbles on vertices $g_r$ and $h_r$ and there are no other pebbles on the graphs. Note that the number of pebbles on each graph is $\lceil r/2\rceil$ and that implicit pebbles (together with the pebbles on $g_r,h_r$) correspond exactly to the pebbles on groups in Version II. Using Lemma~\ref{graph pebble game lemma} we can assume throughout the game that Duplicator chooses bijections on graphs that restrict to bijections on groups and that those restrictions respect implicit pebbles or otherwise Spoiler would win Version III right away. That means that given an admissible configuration, the bijection Duplicator chooses in Version III can be used as a Version II bijection as well. Spoiler will then move in Version II and we argue that Spoiler can win in Version III or force Duplicator into another admissible configuration. Since Spoiler can choose arbitrary moves in Version II, by assumption Spoiler will win in Version II at some point. Using Lemma~\ref{graph pebble game lemma} again, we see that Duplicator will eventually lose the Version III game on this configuration. It remains to argue that Spoiler can maintain admissible configurations. Assume Duplicator chose a bijection $f\colon \Gamma_G\to\Gamma_H$ as above. There are two cases: Spoiler moves a pebble or introduces a new one. Suppose Spoiler introduces a new pebble pair on $g_{r+1}$ and $h_{r+1}$ in the groups. If~$r+1$ is even the new pebble on $g_{r+1}$ is grouped with the already existing pebble on $g_r$. In the graph, Spoiler will put a pebble on $M(g_r,g_{r+1})$. The corresponding pebble should be put on $M(h_r,h_{r+1})$ to obtain an admissible configuration. But since $g_{r+1}$ was not necessarily (implicitly) pebbled it may be the case that $f$ does not map $M(g_r,g_{r+1})$ to $M(h_r,h_{r+1})$. To fix this, Spoiler first pebbles $g_{r+1}$ and $h_{r+1}$ directly, using one additional pebble, and asks for another bijection. By Lemma~\ref{graph pebble game lemma} this bijection must now map $M(g_r,g_{r+1})$ to $M(h_r,h_{r+1})$ and Spoiler reaches an admissible configuration in Version III in two more moves, removing the additional pebble again. The case where a pebble is moved rather than newly introduced can be treated in the same way.
\end{proof}

\begin{lem}
	If $G\not\equiv_{WL_k^\versthree} H$ then $G\not\equiv_{WL_{2k+1}^\versone} H$.
\end{lem}
\begin{proof}
We now want to reverse the argument from the last lemma. Given a pebble-configuration on graphs $\Gamma_G$ and $\Gamma_H$ we call a pebble-configuration on groups \emph{admissible} if the following holds: for each pair of pebbles on element-vertices the corresponding elements in $G$ and $H$ are pebbled as well and for each pair of pebbles on non-element vertices, that is, pebbles on gadgets $M(g_1,g_2)$ and $M(h_1,h_2)$, there are pairs of pebbles on $g_1$ and $h_1$, $g_2$ and $h_2$, respectively. Note that, w.l.o.g., (non-)element-vertices are only pebbled among each other by Lemma \ref{graph pebble game lemma}. Also the number of pebbles on groups in an admissible configuration is at most twice the number of pebbles on graphs. Since Spoiler can move two implicit pebbles at once in Version III we will look at two consecutive rounds of the Version I game at once. Let Duplicator choose a bijection $\varphi\colon G\to H$ in the Version~I game. For each possible Spoiler move introducing an additional pebble on $x$, Duplicator has to commit to one bijection $\varphi_{x}$ on the new configuration. We can force Duplicator to choose this bijection in the corresponding configuration from now on without changing the deterministic outcome of the game, because Duplicator is allowed to choose it freely once. This gives rise to a bijection between pairs of group elements mapping $(x,y)$ to $(\varphi(x),\varphi_x(y))$. Note that this happens without actually making moves, rather think of Duplicators strategy as being precomputable by Spoiler due to the deterministic nature of the game. The map on pairs can now be interpreted as a mapping between element vertices together with a mapping of corresponding multiplication gadgets and will be used as the next Duplicator move in the Version III game. If the Spoiler move in Version III moves two implicit pebbles at once, Spoiler can reach an admissible configuration in three rounds (one additional round for discarding the additional pebble) in the Version I game while Duplicator chooses bijections according to the precomputed strategy.

Finally, before Spoiler wins in Version III, Spoiler will win in Version I. More precisely, as long as the map on pebbles in the Version I game is multiplicative, the corresponding map induced on the pebbled subgraph will be a graph isomorphism, since multiplicativity on pebbles can be expressed equivalently in terms of mapping gadgets accordingly. 
\end{proof}

\begin{proof}[Proof of Theorem~\ref{thm:relationship:between:versions}]
	The first inclusion is clear. The other inclusions are the content of the previous lemmas.
\end{proof}

We remark that the additive constants in the Theorem~\ref{thm:relationship:between:versions} could be improved for $k>2$ by reusing pebbles, but we do not worry about explicit constants at this point. It is also possible to show $WL_k^\verstwo\preceq WL_{k+1}^\versone$ directly.

\section{Embedding graphs into finite groups}\label{sec:embed:graphs:to:groups}

Next, we describe a construction of finite groups from graphs such that structural properties of the resulting groups are primarily determined by the graphs. We will make this statement more precise in the following.
 From now on fix an odd prime~$p$.

\begin{dfn}
	For each natural number $n$ there is a relatively free group of 				exponent $p\neq 2$ and (nilpotency) class $2$ generated by $n$ elements.
	It admits a finite presentation
	\[
		F_{n,p}=\langle x_1,\dots,x_n\mid R(p,n)\rangle
	\]where $R(p,n)$ consists of the following relations: 
	\begin{enumerate}
		\item For all $1\leq i\leq n$ there is a relation $x_i^p=1$, and
		\item for all $1\leq i,j,k\leq n$ there is a relation
		$[[x_i,\!x_j],\!x_k]=1$.
	\end{enumerate}
\end{dfn} 
Thus, the group is generated by~$x_1,\ldots,x_n$, each of these generators is an element of order~$p$, and the commutator of two generators commutes with every generator and thus every element of the group. 
It follows from these properties that elements of $F_{n,p}$ can be uniquely written  as
\[
	x_1^{d_1}\cdot\ldots\cdot x_n^{d_n}
	[x_1,x_2]^{d_{1,2}}[x_1,x_3]^{d_{1,3}}\cdot\ldots\cdot
	[x_{n-1},x_n]^{d_{n-1,n}}
\]
where exponents are defined modulo $p$. In particular, $|F_{n,p}|=p^{n+n(n-1)/2}$.

The main goal is to construct quotients of $F_{n,p}$ using graphs on vertex set~$\{1,\ldots,n\}$ as templates in a way that translates combinatorial similarity of the graphs (with respect to Weisfeiler-Leman-refinement) to similar subgroup profiles. We will see that this affects other isomorphism invariants as well.

\begin{dfn}
To each (simple, undirected) graph $\Gamma=\left(\{v_1,\dots,v_n\},E\right)$ and prime number $p$ we assign a finite exponent $p$ group of nilpotency class $2$ via
\[
	G_\Gamma := \left\langle x_1,\dots,x_n\mid R(p,n), 
	[x_i,x_j]=1: \{v_i,v_j\}\in E \right\rangle.
\] 
Thus, in~$G_\Gamma$ two generators~$x_i,x_j$ commute, if the corresponding vertices form an edge in~$\Gamma$.
We usually identify $x_i$ with $v_i$ and use the latter to refer to the vertex as well as the respective element of $G_\Gamma$. 
We fix an order on generators $v_1,\dots,v_n$ and call these the \textit{standard generators} for $G_\Gamma$. The particular presentation 
above is called the \textit{presentation of $G_{\Gamma}$ from $\Gamma$}.
\end{dfn}

It turns out that this construction has also been used in other contexts. It is sometimes called \textit{Mekler's construction} in the literature (see \cite{mekler_1981} for Mekler's original work) and has been primarily investigated for infinite graphs with respect to model theoretic properties. We first collect some possibly well known combinatorial and group theoretic properties.

\begin{lem}\label{lem:vertices:are:gen:set}
	We have $\Phi(G_\Gamma)=G'_\Gamma$ and the vertices of $\Gamma$ form a generating set of $G_\Gamma$ of minimal cardinality.
\end{lem}

\begin{proof}
By construction $G_\Gamma$ has exponent $p$ and thus $\Phi(G_\Gamma)=G'_\Gamma$ (since for~$p$-groups the Frattini-subgroup is the minimal subgroup with elementary abelian quotient).
	The cardinality of a minimal generating set of $G_\Gamma$ is the dimension of
	the $\mathbb{F}_p$-space $G_\Gamma/\Phi(G_\Gamma)$ which is now equal to $G_\Gamma/G'_\Gamma$. We have
	\[
		G_\Gamma/G'_\Gamma\cong \left\langle V(\Gamma)\mid \text{exponent $p$, abelian}\right\rangle\cong \mathbb{F}_p^{|V(\Gamma)|}
	\] showing the claim.
\end{proof}

\begin{lem}\label{lem:non-edges:are:basis:of:commutator:subgroup}
	Denote by $d$ the number of non-edges in $\Gamma$. Then $G'_\Gamma\cong\mathbb{F}_p^d$, i.e., the set of non-edges of $\Gamma$ forms a basis in $G'_\Gamma$.
\end{lem}
\begin{proof} 
	We have $G'_\Gamma=(F_{n,p}/N)'$ for some normal subgroup $N\leq F_{n,p}'$ with $|N|=p^{|E(\Gamma)|}$ and since commutators are central in $F_{n,p}$ we have $(F_{n,p}/N)'=F_{n,p}'/N$ where
	$|F_{n,p}'|/|N|=p^{{{n}\choose{2}}-|E(\Gamma)|}=p^d$.
\end{proof}
 This also gives us normal forms for elements of $G_\Gamma$.

\begin{cor}
	Let $\Gamma$ be a (simple) graph. Then we have $|G_\Gamma|=p^{|V(\Gamma)|+\left\vert{{V}\choose{2}}-E(\Gamma)\right\vert}$.
	In particular, every element of $G_\Gamma$ can be written in
	the form 
	\[
		v_1^{d_1}\dots v_n^{d_n}c_1^{d_{n+1}}\dots c_k^{d_{n+k}}
	\]
	where $\{c_1,\dots,c_k\}$ is the set of non-trivial 				commutators 	between generators (i.e., the non-edges of the graph~$\Gamma$) and each $d_i$ is uniquely determined modulo $p$.
\end{cor}

We will see that a lot of information on commutation and centralizers can be deduced from $\Gamma$ directly.  We first need to recall some well known properties of commutators in (nilpotent) groups. 

\begin{lem}[Commutator relations]\label{lem:commutator:rels}
	Let $G$ be a group of nilpotency class $2$. Then
	for all $a,b,c\in G$ we have
	\begin{enumerate}
		\item $[a,b]=[b,a^{-1}]$ and
		\item $[a,bc]=[a,b][a,c]$.
	\end{enumerate}
	In particular for all $n,m\in \mathbb{N}$  we have~$[a^m,b^n]=[a,b]^{mn}$. 
	
\end{lem}

\begin{proof}  
	Recall that nilpotency class $2$ means that all commutators are central in $G$. We thus have 
	$[a,b]=aba^{-1}b^{-1}=aba^{-1}b^{-1}aa^{-1}=a[b,a^{-1}]a^{-1}=[b,a^{-1}]$ and we have that $[a,bc]=abca^{-1}c^{-1}b^{-1}=abca^{-1}c^{-1}aa^{-1}b^{-1}=ab[c,a^{-1}]a^{-1}b^{-1}=[a,b][c,a^{-1}]=[a,b][a,c].$ By induction $[a,b^n]=[a,b]^n$. Finally,
	$[a^m,b^n]=[a^m,b]^n=[b,a^{-m}]^n=[b,a]^{-mn}=[a,b]^{mn}.$
\end{proof}

\begin{lem}\label{lem:structure:of:center}
	We have $Z(G_\Gamma)=G'_\Gamma\times \langle v : N[v]=V(\Gamma)\rangle$.	
	In particular, if no vertex of $\Gamma$ is adjacent to all other vertices then $Z(G_\Gamma)=G'_\Gamma$.
\end{lem}

\begin{proof} 
	We can assume that no vertex in $\Gamma$ is adjacent to all other vertices. Now take an arbitrary element $x:=v_1^{d_1}\dots v_n^{d_n}c_1^{d_{n+1}}\dots c_k^{d_{n+k}}$ like above. If $d_i$ is non-trivial modulo $p$ for some $i\leq n$ then by assumption we find some vertex $v_j$ such that
	$[v_i,v_j]$ is non-trivial. By the counting argument above, commutators of different pairs of generators are linearly independent and using commutator relations we see that thus $[x,v_j]$ is non-trivial as well. So either $d_i\equiv_p 0$ for all $i\leq n$ and $x$ is a product of commutators, or $x$ is not central.
\end{proof}

From now on let us fix a graph $\Gamma$ on vertex set $\{v_1,\dots,v_n\}$ and let $G := G_\Gamma$. We set $m := |\Phi(G)|
=|G'|$. Then $m$ is the number of non-edges in $\Gamma$ and
$|G|=p^{m+n}$. Furthermore, fix an ordering of non-trivial commutators $c_1,\dots,c_m$ of pairs of standard generators~$[v_i,v_j]\neq 1$ with~$i<j$.

\begin{dfn}
	Let $x\in G_\Gamma$ be an element with normal form
	\[
		x := v_1^{d_1}\dots v_n^{d_n}c_1^{e_1}\dots c_m^{e_m}
	\]The \textbf{support} of $x$ is $\{ v_i\mid d_i\not\equiv_p 0\}$. For a subset of vertices $S\subseteq V(\Gamma)$ let 
	$x_S$ be the subword $v_{i_1}^{d_{i_1}}\dots v_{i_s}^{d_{i_s}}$ where $S=\{ v_{i_1},\ldots,v_{i_s}\}$ with  $i_1<\dots<i_s$.
\end{dfn}

Towards analyzing commutation in~$G_\Gamma$ we consider an example.

\begin{exm}\label{example:centralizer}
Note that for two connected components $C_1,C_2$ of the complement graph $\co(\Gamma)$ and every 
group element $x\in G$ we always have $x_{C_1}x_{C_2}
=x_{C_2}x_{C_1}$.
	Consider the complete bipartite graph~$\Gamma$ on parts~$\{v_1,v_2\}$ and~$\{v_3,v_4\}$ and its complement~$\co(\Gamma)$ (see Figure~\ref{fig:example:graph}).
	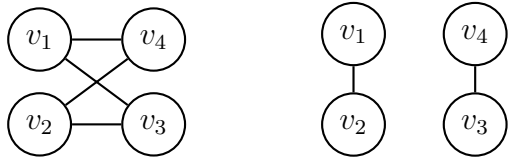
\begin{figure}
	\centering
			\begin{tikzpicture}[thick,scale=0.75]
			\tikzstyle{every node}=[draw,shape=circle]
			\path 
			(-1,-1) node (p1) {$v_1$}
			(-1,-2.5) node (p2) {$v_2$}
			(1,-2.5) node (p3) {$v_3$}
			(1,-1) node (p4) {$v_4$};
			\draw
			(p1) -- (p3)
			(p1) -- (p4)
			(p2) -- (p3)
			(p2) -- (p4);
		\end{tikzpicture}
		\hspace{1.5cm}
		\begin{tikzpicture}[thick, scale=0.8]
			\tikzstyle{every node}=[draw,shape=circle]
			\path 
			(-1,-1) node (p1) {$v_1$}
			(-1,-2.5) node (p2) {$v_2$}
			(1,-2.5) node (p3) {$v_3$}
			(1,-1) node (p4) {$v_4$};
			\draw
			(p1) -- (p2)
			(p3) -- (p4);
		\end{tikzpicture}
		\caption{A complete bipartite graph on 4 vertices (left) and its complement (right).\label{fig:example:graph}}
		\end{figure}
		Then
		\[
			C_{G_\Gamma}(v_1v_2v_3v_4)=\langle v_1v_2\rangle
			\langle v_3v_4\rangle
			 Z(G_\Gamma).
		\] 
			
\end{exm} 

The following theorem states that the example essentially captures how commutation works in general.

\begin{lem}\label{CommutationInComponents}
For~$x\in G_\Gamma$ let $C_1,\dots,C_s$ be the connected components of $\co(\Gamma[\supp(x)])$. Then~$x=x_{C_1}\cdots x_{C_s}c$ with $c\in G'_\Gamma\leq Z(G_\Gamma)$ and~$y\in G_\Gamma$ commutes with~$x$ if and only if~$y \in\langle x_{C_1}\rangle\cdots \langle x_{C_s}\rangle\cdot\langle w: [v,w]=1\text{ for all }v\in \supp(x)\rangle G'_\Gamma$. 
\end{lem}

\begin{proof} 
	By definition of $G_\Gamma$ for $i\neq j$ all elements belonging to $C_i$ commute with all elements from $C_j$, giving rise to a decomposition of $x$ into parts belonging to components of $\co(\Gamma)$. Furthermore, it shows that commutation of group elements $x$ and $y$ is the same as simultaneous commutation with all of the respective parts. Consider now the case $x=x_{C_i}$ for some $i$. If $v\in \supp(x)\setminus \supp(y)$ then, due to commutators being independent, $[x,y]=1$ if and only if $v$ commutes with every element from $\supp(y)$ and the same holds after interchanging roles of $x$ and $y$. Thus, we can reduce to the case that $\supp(x)=\supp(y)$ and we will argue that $x$ and $y$ are powers of each other or trivial. For ease of notation assume that $x=v^{d_1}_1\dots v^{d_r}_r$ and $y=v^{f_1}_1\dots v^{f_r}_r$ where $d_i$ and $f_i$ are non-zero modulo $p$. Using commutator relations we obtain
	\[
	[x,y] = [v_1,v_2]^{d_2f_1-d_1f_2}\dots [v_{r-1},v_r]^{d_rf_{r-1}-d_{r-1}f_3}
	\]and for $[x,y]$ to vanish, all of these exponents have to be divisible by $p$. That is, modulo $p$, $f_2$ is uniquely determined by $d_1,d_2$ and $f_1$ or $[v_1,v_2]$ is trivial. Since all $v_i$ lie in one connected component of  $\co(\Gamma[\supp(x)])$, there is a sequence of non-edges from $v_1$ to every $v_i$ within the component and it follows in an inductive fashion that the values of $d_1,\dots,d_r$ together with a choice of $f_1$ uniquely determine all other values of the $f_i$ (modulo $p$). Now clearly one admissible system of exponents is given by choosing $y$ as a power of $x$ and due to uniqueness these are the only possible configurations.
\end{proof}

\begin{cor}\label{centralizer}
	Let $x=v^{d_1}_{i_1}\dots v_{i_r}^{d_r}c$ with $i_1<i_2<\dots <i_r$, $c$ central in $G_\Gamma$ and $d_i \not\equiv_p 0$ for all $i$. Then 
	\[
		C_{G_\Gamma}(x)=\langle x_{C_1}\rangle
		\dots \langle x_{C_s}\rangle \left\langle 
		\{ v_m\mid [v_m,v_{i_j}]=1\text{ for all }j\}\right
		\rangle G'_\Gamma.
	\]
	Where, $C_1,\dots,C_s$ are the connected components of the complement graph~$\co(\Gamma[\supp(x)])$.
\end{cor}

This (almost) distinguishes single support vertices.

\begin{lem}\label{lem:centralizerComparison}
For $x\in G_\Gamma$ and $v\in \supp(x)$ we have that $|C_{G_\Gamma}(x)|\leq |C_{G_\Gamma}(v)|$. Set $M(x):=\{w\in V(\Gamma)\mid [w,y]=1\text{ for all }y\in \supp(x)\}$. Then if $|C_{G_\Gamma}(x)|=|C_{G_\Gamma}(v)|$ either
$M(x)=M(v)$ in which case $\Gamma[\supp(x)]$ is a complete graph, or
$M(x)=M(v)\setminus\{v\}$ and in both cases all components of $\co(\Gamma[\supp(x)])$ not containing $v$ are singletons.
\end{lem}

\begin{proof}
Write $x=v^{d_1}_{i_1}\dots v_{i_r}^{d_r}c$ and
\[
	C_{G_\Gamma}(x)=\langle x_{C_1}\rangle
		\dots \langle x_{C_s}\rangle \left\langle 
		\{ v_m\mid [v_m,v_{i_j}]=1\text{ for all }j\}\right
		\rangle G'_\Gamma
\]as above. Assume, w.l.o.g., that $v$ is contained in the component $C_1$ of $\co(\Gamma[\supp(x)])$. Then clearly $x_{C_2},\dots,x_{C_s}\in C_{G_\Gamma}(v)$ and whenever
$[v_m,v_{i_j}]=1\text{ for all }j$ then $[v_m,v]=1$ in particular. Both $C_{G_\Gamma}(x)$ and $C_{G_\Gamma}(v)$ contain $G'_\Gamma\leq Z(G_\Gamma)$ and form $\mathbb{F}_p$-spaces modulo $G'_\Gamma$. Thus $|C_{G_\Gamma}(x)|\leq |C_{G_\Gamma}(v)|$ is equivalent to $\dim_{\mathbb{F}_p}(C_{G_\Gamma}(x)/G'_\Gamma)\leq \dim_{\mathbb{F}_p}(C_{G_\Gamma}(v)/G'_\Gamma)$.

Now $C_1,\dots,C_s$ partition $\supp(x)\subseteq V(\Gamma)$ and $V(\Gamma)$ is linearly independent modulo $G'_\Gamma$ by definition of $G_\Gamma$. Assume $w\in M(x)\cap C_i$ for some $i$ then $w$ commutes with all vertices from $\supp(x)$ and this is equivalent to $C_i=\{w\}$. So $C_{G_\Gamma}(x)/G'_\Gamma$ has a basis of the form $\{x_{C_i}G'_\Gamma\mid |C_i|>1\}\cup\{wG'_\Gamma\mid w\in M(x)\}$ and these sets are disjoint. Now we always have $M(x)\subseteq M(v)$ and for $i>1$ it holds $C_i\subseteq M(v)$ (so in particular $x_{C_i}\in M(v)$ as well). If $|C_1|=1$ (so $C_1=\{v\}$) then $\{x_{C_i}G'_\Gamma\mid |C_i|>1\}\cup\{wG'_\Gamma\mid w\in M(x)\}$ is completely contained in $C_{G_\Gamma}(v)$. If $|C_1|>1$ then $v\notin M(x)$ and by the argument above $\{x_{C_i}G'_\Gamma\mid i>1,|C_i|>1\}\cup\{wG'_\Gamma\mid w\in M(x)\}\cup\{vG'_\Gamma\}$ is a union of disjoint sets which is linearly independent modulo $G'_\Gamma$. In both cases $|C_{G_\Gamma}(x)|\leq |C_{G_\Gamma}(v)|$ and if $|C_i|>1$ for some $i>1$ then actually we get a proper inequality (all elements from $C_i$ contribute to $\dim_{\mathbb{F}_p}(C_{G_\Gamma}(v)/G'_\Gamma)$ separately).
So if equality holds then $\dim_{\mathbb{F}_p}(C_{G_\Gamma}(x)/G'_\Gamma)\leq M(x)+1$ (since all $C_i$ apart from maybe $C_1$ are covered by $M(x)$) and assuming $M(x)\neq M(v)$ we additionally must have $|M(v)|=|M(x)|+1$ showing that in this case $|C_1|>1$ and $v\notin M(x)$.
\end{proof} 

This means that elements of the form $vz$ with $v\in V(\Gamma)$ and $z\in G'_\Gamma$ are almost canonical in $G_\Gamma$ in the following sense: Define a set $\mathcal{C}$ as the union of all minimal generating sets $\{g_1,\dots,g_n\}$ of $G_\Gamma$ (so $n=|V(\Gamma)|$) for which the value of $\sum_i |C_{G_\Gamma}(g_i)|$ is maximal among minimal generating sets of $G_\Gamma$. Then $\mathcal{C}$ contains $V(\Gamma)$ since $V(\Gamma)$ is such a generating set itself. Furthermore $\mathcal{C}$ is canonical in $G_\Gamma$ (invariant under all automorphisms) and we can use it to analyze isomorphisms.

In the following part we want to compare different groups presented on graphs. Let us fix graphs $\Gamma_1$ and $\Gamma_2$ on the vertex set $\{v_1,\dots,v_n\}$ with edges given by $E_1$ and $E_2$ and corresponding groups $G_i:=G_{\Gamma_i}$. The standard generators on which the $G_i$ are presented will again be called $(v_j)_{1\leq j\leq n}$. 

\begin{thm}
	It holds that $\Gamma_1\cong \Gamma_2$ if and only if $G_1\cong G_2$.
\end{thm}	
\begin{proof}
	Let $\varphi\colon\Gamma_1\to \Gamma_2$ be a graph isomorphism. Then $\varphi$ induces an automorphism of $F_{n,p}$ by permuting generators and we have 
	$G_{\Gamma_i}=F_{n,p}/N_i$ where $N_i$ is the central subgroup generated by edges of $\Gamma_i$. Thus, as a group automorphism, $\varphi$ maps $N_1$ to $N_2$ giving an isomorphism of the corresponding quotients. 
	
	For the other direction consider a group isomorphism $\varphi:G_1\to G_2$. From Lemma~\ref{lem:centralizerComparison} we see that for $x\in G_i$ and $v\in \supp(x)$ we have
	\[
		\circledast:\ |C_{G_i}(x)|\leq |C_{G_i}(v)|.
	\]
As in the last lemma let $M(x):=\{v\in V(\Gamma_i)\mid [v,w]=1\text{ for all }w\in \supp(x)\}$ for $x\in G$ be the set of standard generators commuting with the entire support of~$x$. In fact~$M(x)= \bigcap_{w\in \supp(x)} N[w]$.

Our strategy is now to alter the group isomorphism $\varphi:G_1\to G_2$ until we can extract sufficiently much information on the graphs. We do so by redefining the images~$y_i=\varphi(v_i)$ and double checking that the new map is still a homomorphism onto a generating set and thus an isomorphism.

Consider the case that $y:=y_j$ is supported in $G_2$ by more than one vertex for some index $j\leq n$. There must be some vertex $v\in \supp(y)$ such that replacing $y$ with $v$ still leaves us with a generating set for $G_2$. Indeed, this is true in the elementary abelian group $G_2/(G_2)'$ and commutators are non-generators in $G_2$. Furthermore, from $\circledast$ it follows that $(v_1,\dots,v_n)$ is a generating set of $G_1$ which maximizes the sum of centralizer orders $\sum_i |C_{G_1}(v_i)|$ among minimal generating sets and since $\varphi$ is an isomorphism, the same must be true for $(y_1,\dots,y_n)$ in $G_2$. 
For $i>1$, consider $y_i$ such that $[y,y_i]=1$. From Corollary \ref{centralizer} we see that (up to multiplication with commutators which can be ignored) $y_i=y_{C_1}^{t_1}\dots y_{C_s}^{t_s} v^{e_1}_{i_1}\dots v^{e_k}_{i_k}$ for some vertices $v_{i_j}\in M(y)$ and where $C_1,\dots C_s$ are the components of $\co(\Gamma[\supp(y)])$ and we also get that
$[y_i,v]=[y_{C_1}^{t_1},v]$ where we, w.l.o.g., assume that $v\in C_1$. The last Corollary furthermore shows that $|C_i|=1$ for $i>1$, so actually we can write $y_i=y_{C_1}^{t_1}v^{e_1}_{i_1}\dots v^{e_k}_{i_{k'}}$ for $v_{i_j}\in M(y)$. Using the same argument as for $y$ and $v$ there is some $w\in \supp(y_i)$ such that $y_i$ can be replaced with $w$ while still keeping a generating set and for this $w$ we again have $|C_{G_2}(y_i)|=|C_{G_2}(w)|$. Also note that if $|C_1|=1$ then $[v,y_i]=1$ which is what we want to show. Similarly we are done if $t_1\equiv_p 0$, so assume otherwise. If $|C_1|>1$ there is some $v'\in C_1$ such that $[v,v']\neq 1$ and in particular $v,v'\notin M(y_i)$ implying that $v,v'\in\supp(y_i)$ from the expression for $y_i$ above. Now $w$ can be chosen such that~$w\notin C_1$ (since the exponents of $y$ and $y_i$ over elements of $C_i$ agree this follows from rank considerations and the fact that~$(y_1G'_2,\dots,y_nG'_2)$ forms a basis of~$G_2/G'_2$). Thus~$M(y_i)\subseteq M(w)\setminus\{v,v'\}$ contradicting the previous Corollary.

In conclusion, $[y,y_i]=1$ implies $[v,y_i]=1$ (And we even see that this only happens if $\supp(y)$ induces a complete graph or if $\supp(y_i)\cap C_1=\emptyset$). Hence exchanging $y$ for $v$ gives us a generating set which is still a valid image of $(v_1,\dots,v_n)$.
We can iterate this process to obtain an isomorphism mapping vertices to elements supported by single vertices as well which gives rise to a bijection between vertices. The fact that the isomorphism respects commutators then translates to respecting edges of the graphs and we conclude that $\Gamma_1\cong\Gamma_2$.
\end{proof}

It is not always the case that the original vertices~$V(\Gamma)$ of the graph form a canonical subset of $G_\Gamma$, even when taken modulo commutators. However, we can precisely describe the conditions under which they do (In a previous version of the paper we neglected the inclusion of commutators in the canonical set. We thank Ilia Ponomarenko for pointing this out to us).

\begin{lem}\label{lem:char:when:iso:pres:red}Assume $\Gamma_1\cong\Gamma_2$.
Then $\Gamma_iG'_i$ is canonical in $G_i$, if and only if in ~$\Gamma_1$ (and thus~$\Gamma_2$) there is no pair of distinct vertices~$v,w$ with~$N(v)\subseteq N[w]$.

In this case, each element of $\Iso(G_1,G_2)$ uniquely determines an element of $\Iso(\Gamma_1,\Gamma_2)$ by restriction to $\{vG'_1\mid v\in\Gamma_1\}$.
\end{lem}

\begin{proof}
	Following the last proof we see that elements with single-vertex support are canonical in $G_1$ and $G_2$ under the condition above. Assume the condition does not hold in $\Gamma_1$ and for distinct vertices $v\neq w$ we have $N(v)\subseteq N[w]$. Then mapping $v$ to $vw$ and fixing other generators extends to an automorphism of $G_1$ via the given presentation of~$G_1$ from~$\Gamma_1$.
\end{proof}

\section[Constructing groups with equal k-profiles]{Constructing groups with equal $k$-profiles}

In this section we want to apply the construction from above to specific graphs. The idea is to start with a family of $3$-regular base graphs such that the CFI-construction gives us two non-isomorphic graphs $\Gamma_1$ and $\Gamma_2$ for each of the base graphs which can be distinguished by $k$-WL only for $k$ scaling linearly with the size of the CFI-graphs. We will then show that the resulting groups $G_i:=G_{\Gamma_i}$ have equal $\Theta(k)$-profiles.

\begin{dfn}
	For a group $G$, a tuple $(g_1,\dots,g_k)\in G^k$ is \textit{minimal} if
	$\langle g_1,\dots,g_k\rangle$ is not generated by $k-1$ elements.
\end{dfn}

When working with $F := F_{n,p}$ we will fix a standard basis for $Z(F)=\Phi(F)\cong\mathbb{F}_p^{{n}\choose{2}}$. If $F$ is presented on generators $v_1,\dots,v_n$ we choose
\[
	([v_1,v_2],[v_1,v_3],\dots,[v_1,v_n],[v_2,v_3],\dots,[v_{n-1},v_n])
\] as our fixed basis for the center of $F$. We call these commutators the \emph{standard commutators}.

\begin{dfn}
	Let $\bar{g}:=(g_1,\dots,g_k)\in F_{n,p}^k$. We define two $\mathbb{F}_p$-matrices.
In the~$(k\times n)$-matrix $B_1(\bar{g})$ the $i$-th row corresponds to (the exponents of) $g_i$ expressed in normal form in terms of standard generators. In the~$\left(\binom{k}{2}\times \binom{n}{2}\right)$-matrix $B_2(\bar{g})$ the rows correspond to
	$[g_1,g_2],[g_1,g_3],\dots,[g_{k-1},g_k]$ expressed in terms of standard commutators in this order. We will sometimes refer to their columns
	by these labels, i.e., the column belonging to $[v_i,v_j]$ will be referenced as $B_2(\bar{g})([v_i,v_j])$.
\end{dfn}

\begin{exm}
For example assume~$n=3$ and~$k=2$ and assume $\bar{g}:=(g_1,g_2)$ with~$g_1=v_1v_2^5v_3$ and~$g_2= v_1^2v_2[v_1,v_2]$. Then~$[g_1,g_2]=$
\[[v_3,v_1^2][v_2^5,v_1^2][v_3,v_2][v_1,v_2]\!=\![v_1,v_2]^{-9}[v_1,v_3]^{-2}[v_2,v_3]^{-1}.\]

In this case $B_1(\bar{g}) = \begin{pmatrix}

1 & 5 &1 \\

2 & 1  &0\\\end{pmatrix}$ and $B_2(\bar{g}) = (-9, -2,-1)$
\end{exm}where entries are to be read modulo $p$.

\begin{lem}\label{ExtPower}
	Let $\bar{g}:=(g_1,\dots,g_t)\in (F_{n,p})^t$. Then
	$B_2(\bar{g})=B_1(\bar{g})\wedge B_1(\bar{g})$ where $\wedge$ describes the exterior product with respect to our chosen orderings for the standard bases.
\end{lem}
\begin{proof}
	Express the commutator $c_{i,j}:=[g_i,g_j]$ in terms of the standard commutators. Then $c_{i,j}=\left([v_k,v_\ell]^{m(k,\ell)}\right)_{k<\ell}$ where~$m(k,\ell)={(B_1(\bar{g}))}_{i,k}{(B_1(\bar{g}))}_{j,\ell}-{(B_1(\bar{g}))}_{i,\ell}{(B_1(\bar{g}))}_{j,k}=$
	\[
		\det\begin{pmatrix}
	{(B_1(\bar{g}))}_{i,k} & {(B_1(\bar{g}))}_{i,\ell}\\ {(B_1(\bar{g}))}_{j,k} & {(B_1(\bar{g}))}_{j,\ell} \end{pmatrix}.
	\] Thus, the row of $B_2(\bar{g})$ belonging to $c_{i,j}$ corresponds to 	the row of $B_1(\bar{g})\wedge B_1(\bar{g})$ belonging to rows ${B_1(\bar{g})}_{i,-}$ and ${B_1(\bar{g})}_{j,-}$ 
\end{proof}

In particular, this shows that subgroups of $F_{n,p}$ are direct products of relatively free groups and central groups. In the following we will use the fact that for $M\in\mathbb{F}_p^{k\times n}$ we have $\rank(M\wedge M)={{\rank(M)}\choose{2}}$,~see for example \cite[Section 10.1]{MR1153019}.

\begin{lem}
	Let $G\leq F_{n,p}$ be generated by $\bar{g}:=(g_1,\dots,g_t)$ and set $r:=\rank(B_1(\bar{g}))$. Then there are $r$ elements $g_{i_j}$ among $\{g_1,\dots,g_t\}$ and 
	central elements $c_1,\dots,c_k\in Z(F_{n,p})$ for some $0\leq k\leq (n-r)$ such that $G=\langle g_{i_1},\dots,g_{i_r}\rangle\times\langle c_1,\dots,c_k\rangle$. 
	Furthermore, $G'$ has $\mathbb{F}_p$-dimension ${{r}\choose{2}}$.
\end{lem}
\begin{proof}
	If $B_1(\bar{g})$ has rank $r$, we can choose $r$ linearly independent rows corresponding to certain generators $g_{i_j}$. Other rows can then be expressed via
	these chosen rows which by definition of $B_1(\bar{g})$ means that all other generators can be replaced by central elements $c_1,\dots,c_{n-r}$ without changing $G$.
	Set $G_r := \langle g_{i_1},\dots,g_{i_r}\rangle$. The corresponding rows in $B_1(\bar{g})$ are now independent meaning that no set of cardinality less than $r$ can generate $G_r$. Since all other generators are now central we have $[G,G]=[G_r,G_r]$ and the latter is of dimension $\rank(B_2(\bar{g}))={{r}\choose{2}}$. Choose a subset of $c_i$'s that is maximal with respect to the property $G_r\cap\langle c_{i_1},\dots,c_{i_k}\rangle=\emptyset$. Then $G=\langle G_r,c_{i_1},\dots,c_{i_k}\rangle$ as desired.
\end{proof}

The following observation is elementary but will help us compare subgroups of $G_\Gamma$ for different values of $\Gamma$.
\begin{lem}\label{subgroupsModN1}
	Let $H:=\langle g_1,\dots,g_t,z_1,\dots,z_r\rangle\leq G_\Gamma$ and $R:=\dim(\Phi(H))$. Assume that all $z_i$ are central in $G_\Gamma$, that $\langle g_1,\dots,g_t\rangle/Z(G_\Gamma)\cong \mathbb{F}_p^t$, and that $H$ is not generated by less than $t+r$ elements. Let $c_1,\dots,c_R$ be generators
	of $\Phi(H)$ of the form $c_i=[g_{i_1},g_{i_2}]$ and express
	all other commutators $c_{R+1},\dots,c_{{t}\choose{2}}$ 			between the $g_i$ as words $w_{R+1},\dots,w_{{t}\choose{2}}$
	in the $c_i$. Then
	\[
		H\cong\langle g_1,\dots,g_t\mid \textup{exponent~$p$, class 2 },
		w_{R+1},\dots,w_{{t}\choose{2}}\rangle\times C^r_p.
	\]
\end{lem}
\begin{proof}
	By assumption $|\langle g_1,\dots,g_t\rangle|=p^{t+R}$. Clearly the presentation above defines a group admitting an epimorphism 
	onto $\langle g_1,\dots,g_t\rangle$. Due to the given relations its order is at most
	$p^{t+R}$. Since $(g_1,\dots,g_t,z_1,\dots,z_r)$ is assumed to be minimal, the central group $\langle z_1,\dots,z_r\rangle\cong C^r_p$ splits from $H$.
\end{proof}

Let $\Gamma_0=(\{V_1,\dots,V_t\},E)$ be a $3$-regular graph with $N:=|E|$ edges and such that $\Gamma_1:=\CFI(\Gamma_0)$ and $\Gamma_2:=\widetilde{\CFI(\Gamma_0)}$ are not isomorphic (cf.~Theorem~\ref{thm:cfi}). Let $n:=10t$ be the number of vertices of $\Gamma_1$ and $\Gamma_2$. 

To improve readability, we use capital letters for the vertices of the base graph in the following.

We assume $\Gamma_0$ and $\co(\Gamma_0)$ to be connected and then the same holds for the corresponding CFI-graphs. Recall that the CFI-graphs are again $3$-regular. In the following we will call a pair of edges between two CFI-gadgets together with their adjacent vertices a \textit{link} and \textit{twisting} will be understood as replacing the edges in a link with their twisted version. Note that two gadgets or two links are always disjoint or equal and that links correspond bijectively to edges in the base graph $\Gamma_0$. As before we call vertices of links external (w.r.t.~their gadget) and other vertices internal. We fix $F:=F(n,p)$, the relatively free group on vertices of the CFI-graphs above. We also fix normal subgroups $N_1,N_2\leq F$ corresponding to edges of $\Gamma_1$ and~$\Gamma_2$, respectively. Thus $G_i:=G_{\Gamma_i}=F/N_i$. Finally, let $e$ be any edge in the base graph and let $\mbox{}^{(e)}\colon F\to F$ be the following map:
Say $e=(V,W)\in E(\Gamma_0)$ (so we actually chose an orientation).  Then twisting along $e$ can be seen as swapping
in all normal forms the standard commutators~$[a^V_i,a^W_j]$ and~$[a^V_i,b^W_j]$ and also swapping all occurrences of~$[b^V_i,b^W_j]$ and~$[b^V_i,a^W_j]$.
This is of course not a group isomorphism but it induces an automorphism $\varphi\colon Z(F)\to Z(F)$. If $x\in F$ has a normal form that factors as $vc$ where $v$ is the part of $x$ in standard generators and $c$ is the product of standard commutators then $x^{(e)}:=v\varphi(c)$ and this defines a bijection of $F$ into itself.

\begin{dfn}
	A group $H\leq F$ is called \textit{essentially $k$-generated} if
	\begin{enumerate}
		\item $F'=Z(F)\leq H$ and
		\item $\dim_{\mathbb{F}_p}(H/F')=k$.
	\end{enumerate}
	Intuitively this means that the group is~$k$ generated modulo the center.
	Define $\mathcal{H}_k\subseteq \Sub(F)$ to be the set of all essentially $k$-generated subgroups of $F$.
\end{dfn}

\begin{lem}\label{lem:all;in:unique:essent:k:gen}
	For every subgroup $S\leq G_i=F/N_i$ for which $\dim_{\mathbb{F}_p}(S/G'_i)=k$ there is a unique essentially $k$-generated subgroup $H\leq F$ such that $S\leq H/N_i$.
\end{lem}
\begin{proof}
	Let $\nu:F\to F/N_i$ be the natural epimorphism then $H$ can be uniquely defined as $\nu^{-1}(S)F'$.
\end{proof}
Set~$\mathcal{H}_k^{N_i}:=\{ H/N_i\mid H\in\mathcal{H}_k\}$. Our goal is for various $k$ to construct a bijection
\[
	\mathcal{H}_k^{N_1}\to\mathcal{H}_k^{N_2}
\]
that preserves isomorphism-types of groups. Since all~$k$-generated subgroups have the property that $\dim_{\mathbb{F}_p}(S/G'_i)=k$,  the lemma above then gives an isomorphism-type preserving bijection between $k$-generated subgroups of $G_1$ and $G_2$. Note that $\mathcal{H}_k^{N_i}=\{ S\leq G_i\mid \dim_{\mathbb{F}_p}(S/G'_i)=k\text{ and }G'_i=Z(G_i)\leq S \}$.

\begin{lem}\label{TwistingSubgroups}
	Let $1\leq k<N/10$ where $N$ is the number of edges in $\Gamma_0$. 		For $H\in\mathcal{H}_k$ there is some edge $e$ in the base graph such that $H/N_1\cong H/N_1^{(e)}$.
\end{lem}
\begin{proof}
	Let $H:=\langle f_1,\dots,f_\ell\rangle$ and for each $i$ set $g_i:=f_iN_1\in G_1$. We want to investigate the group $(H/N_1)'=H'/N_1$. Since it is generated by commutators between the $g_i$ its structure is mostly described by $B_2(\bar{f})=B_1(\bar{f})^{\wedge 2}$ after replacing columns indexed by elements of $N_1$ with zero-columns. Call this new matrix $B_2(\bar{g})$.

Twisting along edge $e=(V,W)\leq E(\Gamma_0)$ results in mapping
$([a^V_i,a^W_j],[b^V_i,b^W_j])$ to $([a^V_i,b^W_j],[b^V_i,a^W_j])$ (and vice versa, see Section~\ref{sec:CFI-graphs}). This can also be interpreted in terms of the matrices from above as replacing the two zero-columns $B_2(\bar{g})([a^V_i,a^W_j])$ and $B_2(\bar{g})([b^V_i,b^W_j])$ by the original columns in $B_2(\bar{f})$ and replacing the columns corresponding to $[a^V_i,b^W_j]$ and $[b^V_i,a^W_j]$ with zero-columns instead. This defines a matrix $B_2(\bar{g}^{(e)})$ that describes linear dependencies between commutators among the $(f_1^{(e)}\dots,f_k^{(e)})$ modulo $N_1^{(e)}$.

We will now argue that $e$ can be chosen in such a way that $B_2(\bar{g})$ and $B_2(\bar{g}^{(e)})$ have the same column spaces. For this, we argue that we can fix a system of columns of rank $r$ in $B_2(\bar{g})$ that does not contain the columns affected by twisting along $e$, then for $e$ as above these columns also form a system of maximal rank in $B_2(\bar{g}^{(e)})$ and thus linear dependency relations for rows of the two matrices are exactly the same. Using Lemma \ref{subgroupsModN1} we see that $H/N_1\cong H/N_1^{(e)}$ for this choice of $e$.
	
By assumption the rank of $B_1(\bar{f})$ is $k$ and $k<N$. We assume w.l.o.g.~that the first $k$ columns of $B_1(\bar{f})$ are linearly independent. Then the same holds for the first ${{k}\choose{2}}$ 			columns in $B_2(\bar{f})=B_1(\bar{f})\wedge B_1(\bar{f})$. Now these columns may not contain a system of full rank anymore in $B_2(\bar{g})$ but they belong to commutators of the form $[i,j]$ for $1\leq i<j\leq k$.
Since $\Gamma_i$ is $3$-regular, for a fixed $i$ at most three of these commutators are contained in $N_1$. Thus the rank of the first ${{k}\choose{2}}$ columns in $B_2(\bar{g})$ is at least ${{k}\choose{2}}-3k$
and we may choose $r'\leq 3k$ additional columns such that they contain a system of full rank together with the first ${{k}\choose{2}}$ columns. Now every such column belongs to a pair of vertices and the number of relevant vertices for the full rank system in total is smaller than $2r'+k\leq 7k < N$ and thus there are still links in $\Gamma_1$ that are not adjacent to any of these vertices. Let us say these links correspond at least to edges $e_1,\dots,e_{N-7k}$. For each of these links there are two zero-columns in $B_2(\bar{g})$ and two columns agreeing with $B_2(\bar{f})$ corresponding to the twisted/non-twisted version of this link. Due to the choice of the edges we can now replace all four of these columns by zero-columns without reducing the rank of the resulting matrix. We will argue that among the edges $e_1,\dots,e_{N-7k}$ there are some edges where twisting also does not change the rank.
	
For this, note that for vertices $v,w$, column $(B_2(\bar{f}))([v,w])$ is a linear combination of columns $(B_2(\bar{f}))([v,y])$ and also a linear combination of columns $(B_2(\bar{f}))([y,w])$ where $y$ runs through the first $k$ columns of $B_1(\bar{f})$ since we assumed the first $k$ columns of $B_1(\bar{f})$ to be linearly independent and since the entries of $B_2(\bar{f})$ are subdeterminants of $B_1(\bar{f})$. Say the first $k$ columns of $B_1(\bar{f})$ correspond to vertices $v_1,\dots,v_k$ in the CFI-graphs.
We say that $1\leq i\leq k$ is \textit{bad} for some link if $v_i$ is adjacent to this link. Since each index is bad for at most three links and $3k<N-7k$, there exist links over the edges $e_1,\dots,e_{N-7k}$ for which no index is bad. For such a link, belonging to edge $e$ say, all columns in the linear combination described above are still present in $B_2(\bar{g})$ and thus the rank of this matrix is the same as for $B_2(\bar{g}^{(e)})$.
\end{proof}

\begin{dfn}
	Set $\mathcal{V}:=V(\Gamma_1)$ and identify $\Sym(\mathcal{V})$ as a subgroup of $\Aut(F)$ in
	the natural way. We set $A$ to be the group of permutations of $\Sym(\mathcal{V})$ that map each gadget to itself with an automorphism.
	(I.e.,~$A$ consists of the graph automorphisms after link edges have been removed.)	
\end{dfn}

Note that the group $A$ is abelian. It is generated by the permutations of $\mathcal{V}$ twisting two incident links in $\Gamma_1$ while permuting the inner vertices of their common gadget accordingly to a graph automorphism of the gadget. 
In particular, $A$ stabilizes all links and gadgets setwise.

If $H\in\mathcal{H}_k$ then for any edge $e$ of $\Gamma_0$ we have $H^{(e)}=H$ (even if $\mbox{}^{(e)}$ is not a group isomorphism).
Lemma~\ref{TwistingSubgroups} shows that for $H/N_1\leq G_1$ there is some edge $e$ of $\Gamma_0$ such that $H/N_1\cong H/N^{(e)}_1$ and by the properties of the CFI-construction the twist~$\mbox{}^{(e)}$ can be altered to become the original twist 
via suitable elements from $A$.
 More precisely, in the situation above there is some $\sigma_e\in A$ (only depending on $e$) such that $H/N^{(e)}_1\cong\sigma_e(H_1)/N_2$. This defines an isomorphism-type preserving map
\[
	\Phi:\mathcal{H}_k^{N_1}\to\mathcal{H}_k^{N_2},\ H/N_1\mapsto \sigma_e(H_1)/N_2,
\]
 where $e$ depends on $H$ and we will show that the edges can be chosen in a way that makes $\Phi$ bijective.

\begin{dfn}
	Let $i\in\{1,2\}$. We say that subgroups $H_1/N_i,H_2/N_i\in\mathcal{H}_k^{N_i}$ are of the same \textit{type} if there is some $\sigma\in A$ such that $H_1=\sigma(H_2)$.
\end{dfn}

An inspection of Lemma~\ref{TwistingSubgroups}'s proof shows the choice of edge~$e$ only depends on the type of the subgroups involved. 
\begin{lem}\label{lem:same:edge:in:type}
	If $k< N/10$, the edge~$e$ in Lemma~\ref{TwistingSubgroups} can be chosen to be the same for all subgroups of a fixed type.
\end{lem}
\begin{proof}
	Since $A$ fixes links setwise, positions where twisting preserves the isomorphism type are the same for groups that get mapped to each other via elements from $A$.
\end{proof}

\begin{lem}
	For each edge~$e$ compatible with Lemma~\ref{TwistingSubgroups}, $\Phi$ maps subgroups of different types to subgroups of different types.
\end{lem}
\begin{proof} 
	Assume that \ref{TwistingSubgroups} gives edges $e_1$ and $e_2$ for
	groups $S,\tilde{S}\leq G_1$. Write $S=H/N_1$, $\tilde{S}=\tilde{H}/N_1$ and assume that
	$\sigma_{e_1}(H)/N_2$ and $\sigma_{e_2}(\tilde{H})/N_2$ have the same type. Then there is some $\sigma\in A$ with
	$(\sigma_{e_2}^{-1}\sigma\sigma_{e_1})(H)=\tilde{H}$ and thus $S$ and $\tilde{S}$ have the same type.
\end{proof}

\begin{lem}
	For a fixed type and a fixed edge $e$ (as in Lemma~\ref{TwistingSubgroups}), $\Phi$ is isomorphism-type preserving and injective. 
\end{lem}
\begin{proof}
	Keep the notation from the last lemma but assume $S_1\neq S_2$ are of the same type. Then $H_1\neq H_2$. Thus $\sigma_e(H_1)\neq\sigma_e(H_2)$ which is equivalent to $\sigma_e(H_1)/N_2\neq\sigma_e(H_2)/N_2$ due to $\sigma_e(H_i)$ containing $Z(F)$ and in particular $N_2$.
\end{proof}

All arguments also work for interchanged roles of $G_1$ and $G_2$. In particular this shows that $|\mathcal{H}_k^{N_1}|=|\mathcal{H}_k^{N_2}|$ for each $k$.
\begin{cor}\label{cor:equal:k-profiles}
	$G_1,G_2$ have equal $k$-profiles for $k< N/10$.
\end{cor}
\begin{proof}
Since bijection~$\Phi$ is isomorphism-type preserving, the collection of subgroups in $\mathcal{H}_k^{N_1}$ is mapped bijectively and isomorphism-type preservingly to $\mathcal{H}_k^{N_2}$. Every~$k$-generated subgroup is contained in a unique factor of an essentially~$k$-generated subgroup (Lemma~\ref{lem:all;in:unique:essent:k:gen}) so this induces a bijection from~$k$-generated subgroups to~$k$-generated subgroups.
\end{proof}

Finally, by the CFI-construction and by 3-regularity of the base graph, $N$ is linear in $n=|V(\Gamma_i)|$, thus $n\in\Theta(\sqrt{\log|G_i|})$.
\begin{cor}
	$G_1$ and $G_2$ have equal $\Theta(\sqrt{\log(n)})$-profiles.
\end{cor}

For the commuting graphs of~$G_1$ and~$G_2$, note that  non-central elements in~$G_1$ that are not powers of one another cannot commute if one of the elements has a support of 4 or larger. Whether the Weisfeiler-Leman algorithm of a particular dimension distinguishes the graphs therefore does not change when restricting the commuting graphs to group elements with support size at most~$3$. 
In particular, the commuting graphs cannot be distinguished by their $\mathcal{O}(\sqrt{\log(n)})$-dimensional Weisfeiler-Leman algorithm.

\section{The Weisfeiler-Leman Dimension of groups constructed from CFI-graphs is 3}\label{sec:wl:dim:3}

In the previous section we constructed groups $G_i:=G_{\Gamma_i}$ based on two CFI-graphs $\Gamma_1$ and $\Gamma_2$.
The groups agree in terms of traditional group theoretical invariants (such as exponent, nilpotency class, and the combinatorics of their conjugacy classes) and also with respect to their $k$ generated subgroups for large~$k$. On first sight this might indicate that these groups should be hard to distinguish by combinatorial means but as we will see in this section their WL-dimension is only $3$. Throughout this section we exclusively use WL-algorithms and pebble games of Version II. The main theorem of this section is the following. 
\begin{thm}\label{LowDimension}
Let~$\Gamma_0$ be a 3-regular connected graph and let $\Gamma_1:=\CFI(\Gamma_0)$ and $\Gamma_2:=\widetilde{\CFI(\Gamma_0)}$ be the corresponding CFI-graphs.
	The $3$-dimensional WL-algorithm distinguishes $G_1$ from $G_2$. If additionally~$\Gamma_0$ has (graph) WL-dimension at most~$3$ then  $3$-dimensional WL-algorithm  identifies $G_1$ as well as~$G_2$. 
\end{thm}

Requiring that~$\Gamma_0$ has WL-dimension at most~3 is not a severe restriction (Observation~\ref{obs:base:graph:low:WL:dim}).
Towards proving the theorem we collect several observations on the pebble game that are particular to the groups arising from CFI-graphs.

\begin{lem}\label{lem:must:respect:support:1}
	For each $k\geq 3$, throughout the $k$-pebble game on $G_1$ and $G_2$ Duplicator has to choose bijections that respect the set of elements with single-vertex support~$\{x\mid |\supp(x)|=1\}$. Moreover $\supp(x)=\supp(y)$ and $|\supp(x)|= |\supp (y)|= 1$ must imply~$\supp(f(x))=\supp(f(y))$.
\end{lem}

\begin{proof} 
To see this, it suffices to realize that centralizers of elements with single-vertex support have a different cardinality than other elements. 
Indeed, since the graphs $\Gamma_1$ and $\Gamma_2$ are 3-regular, by Corollary~\ref{centralizer} each single support vertex has a centralizer of cardinality~$p^4 |Z(G_i)|$. However, since~$\co(\Gamma)$ is connected, has no triangles and no cycles of length 4, other elements have a centralizer of cardinality at most~$p^3 |Z(G_i)|$. 

To see the second part of the theorem, note the following: for two elements~$x,y$ with~$|\supp(x)|= |\supp (y)|= 1$ we have~$\supp(x)= \supp(y)$ exactly if~$C(x)=C(y)$. Since commutation and support sizes must be respected this shows the lemma.
\end{proof}

\begin{lem}\label{ElementSupport}
	If Duplicator does not respect support sizes at some point then Spoiler can win with three pebbles.
\end{lem}

\begin{proof} 
	Assume Duplicator chooses a bijection $f\colon G_1\to G_2$ during the $k$-pebble game with $k\geq 4$ such that $|\supp(x)|\neq |\supp(f(x))|$ for some $x\in G_1$. We already discussed that Spoiler has a winning strategy in this situation in the case that one of the supports has cardinality at most $1$. Since the distribution of support sizes in $G_1$ and $G_2$ is the same there is some $x\in G_1$ with
	$|\supp(f(x))|>|\supp(x)|>1$. We can choose some $v_i\in V(\Gamma_1)$ and a natural number $m$ such that $x':=xv_i^m$ has strictly smaller support than $x$. Now $f(v_i^m)$ must also be supported by exactly one element, or otherwise Duplicator loses anyway. Using $4$ pebbles, Spoiler can force Duplicator to map $x'$ to $f(x)f(v_i^m)$.
	Thus, after three additional rounds, the support of $f(x')$ is still strictly bigger than $\supp(x')$ and the result follows by induction.
	\end{proof}

\begin{lem}\label{ElementType}
	For each $k\geq 4$, throughout the $k$-pebble game on $G_1$ and $G_2$ Duplicator has to choose bijections respecting internal vertices and gadgets of the underlying CFI-graphs. Here, elements corresponding to a gadget vertex $v$ are all elements of $vZ(G_i)$. Moreover pairs of vertices lying in a common gadget have to be mapped to pairs in a common gadget.
\end{lem}

\begin{proof}
By Lemma~\ref{lem:must:respect:support:1} the bijection chosen by Duplicator induces a permutation of the vertices~$V(\Gamma_1)$.
By Construction, the CFI-graphs $\Gamma_i$ have the property that every 6-cycle and every 8-cycle runs entirely within one gadget. Moreover every pair of vertices lying in a common gadget lies on a common 6-cycle or on a common 8-cycle.
This implies that Duplicator has to map vertices~$v,w$ in a common gadget to vertices in a common gadget (and vice versa). Indeed, otherwise 
Spoiler can show that $v$ and $w$ are contained in a small cycle but $f(v)$ and $f(w)$ are not (and vice versa).
This in turn implies that Duplicator has to map internal vertices to internal vertices, because internal vertices are not adjacent to vertices in another gadget, but external vertices are.
\end{proof}

Using these observation we can finally prove the theorem.

\begin{proof}[Proof of Theorem \ref{LowDimension}]

	We first define a set $\mathcal{V}$ of special vertices in $\Gamma_1$: For each gadget put exactly one internal vertex in $\mathcal{V}$ and add all adjacent external vertices. Let $v\in G_1$ denote the ordered product of all vertices in $\mathcal{V}$. By Lemma \ref{ElementSupport} Duplicator must choose a bijection for which $f(v)$ has the same support size as $v$. Spoiler puts a pebble on $v$. The Lemma furthermore shows that all future bijections have to map $\supp(v)$ to $\supp(f(v))$ or otherwise Spoiler can pebble some $v_i\in \supp(v)$ with $f(v_i)\notin \supp(f(v))$ and Duplicator will not be able to respect support sizes from here on. Using Lemma \ref{ElementType} we see that $\supp(f(v))=:\mathcal{V}'$ has to be composed exactly as $\supp(v)=\mathcal{V}$, that is, $\mathcal{V}'$ can also be constructed by choosing set of internal vertices, one per gadget, and adding all their adjacent external vertices. The set~$\mathcal{V}$ induces a subgraph of $\Gamma_1$ and similarly~$\mathcal{V'}$ induces a subgraph of $\Gamma_2$.
	We argue these subgraphs have a different number of edges modulo $2$. 
	For this observe the following: if we alter $\mathcal{V}$ by replacing one internal vertex with another one in the same gadget, this changes exactly two neighbors among the external vertices. The new induced subgraph differs then in exactly two locations of two different links. Thus the number of edges in the induced subgraph remains the same modulo $2$. By induction this is true for all possible choices of $\mathcal{V}$. We can thus assume that~$\mathcal{V}=\mathcal{V}'$. However, this implies that $\Gamma_1[\mathcal{V}]$ and  $\Gamma_2[\mathcal{V}']$ disagree in exactly one edge, namely at the twisted link. This shows the graphs have a different number of edges modulo $2$.

	However, we already argued that Duplicator has to map $\mathcal{V}$ to $\mathcal{V}'$. Since the number of edges of $\Gamma_1[\mathcal{V}]$ and  $\Gamma_2[\mathcal{V}']$ disagree, for any suitable bijection some vertex is mapped to a vertex of incorrect degree, which can be exploited by Spoiler. This shows $G_1$ can be distinguished from $G_2$. 
	
	Assume now that additionally the base graph~$\Gamma_0$ has Weisfeiler-Leman dimension at most 3. Suppose that $G$ is any group with $|G|=|G_1|$ that is indistinguishable from~$G_1$. The vertices of $\Gamma_1$ form a canonical copy of $\Gamma_1$ inside of $G_1$ (up to central elements), so there must be a corresponding set in $G$ as well. If the induced commutation graph $\Gamma$ on this set is distinguishable from $\Gamma_1$ then $G_1$ is distinguishable from $G$. From the commutation graph, we can reconstruct a  corresponding base graph~$\Gamma$. Which must be indistinguishable by 3-WL from~$\Gamma_0$. This means it is isomorphic to~$\Gamma_0$ since its Weisfeiler-Leman dimension is at most 3.
	Thus $\Gamma$ is isomorphic to $\Gamma_1$ or $\Gamma_2$. This gives a presentation of $G$ isomorphic to a presentation of $G_1$ or~$G_2$.
\end{proof}

\section{Conclusion}

We defined several versions of the Weisfeiler-Leman algorithm for groups and showed their dimension concepts are linearly related. We then gave a construction of groups from graphs that preserves isomorphism. We can recover combinatorics of the original graph by analyzing commutation in the groups. This allowed us to construct pairs of non-isomorphic groups $G_1,G_2$ with the same~$k$-subgroup profiles for~$k\in\Theta(\sqrt{\log |G_i|})$. These groups are nevertheless identified by the 3-dimensional Weisfeiler Leman algorithm.

The strategy in the pebble we employed to show this exploits that by pebbling products of standard generators Spoiler can essentially force Duplicator to fix an arbitrary set of standard generators simultaneously. Abstractly, on graphs or groups, one could define a new pebble game where pebbles can be placed on sets of vertices. Spoiler now additionally wins the game if the subset relations between pebbled sets disagree in the two structures. This game corresponds to a monadic second order logic where there is still a bound on the number of variables that may be used. While they seem to resemble each other, we are not sure what the precise relationship between this game and the pebble games on groups is. It also seems to be unknown what the expressive power of this game (or the corresponding logic) is.

On another note, it also remains a central open question whether $k$-WL solves the Group Isomorphism Problem for some constant dimension $k$. While a positive answer would place Group Isomorphism in polynomial time, a negative answer would prove the existence of groups of unbounded Weisfeiler-Leman dimension which would provide interesting examples of groups that are even harder to distinguish.

\bibliographystyle{plain} 
\bibliography{../refs}

\end{document}